\documentclass[aip,
graphicx, 
jcp,
reprint,
]{revtex4-2}

\usepackage{setspace}

\usepackage[main=english]{babel}
\usepackage[utf8]{inputenc}
\usepackage{csquotes}

\usepackage{float}
\usepackage{enumitem}
\usepackage{comment}
\usepackage{bm}

\usepackage{graphicx}

\usepackage{amsmath}
\usepackage{amsfonts}
\usepackage{amssymb}
\usepackage{mathrsfs}
\usepackage{mathtools}
\usepackage{stmaryrd}

\usepackage{amsthm}

\newtheorem{proposition}{Proposition}
\theoremstyle{definition}

\newcommand{\diff}{\mathop{}\!\mathrm{d}} 
\newcommand{\difft}{\frac{\diff }{\diff t}}

\newcommand{\hilbert}{\mathcal{H}}
\newcommand{\daggah}{^\dagger}
\usepackage{dsfont}
\newcommand{\Id}{\mathds{1}}

\DeclareMathOperator{\Tensor}{\otimes}

\DeclarePairedDelimiter{\bra}{\langle}{|}
\DeclarePairedDelimiter{\ket}{|}{\rangle}
\newcommand{\ketbra}[2]{\ket{#1}\bra{#2} }

\renewcommand{\Im}{\mathrm{Im}}

\usepackage{blindtext}
\usepackage{hyperref}
\usepackage{cleveref}

\usepackage{xspace}
\newcommand{\schr}{Schr\"{o}dinger\xspace}
\newcommand{\ito}{It\={o}\xspace}

\usepackage{xcolor}

\DeclareUnicodeCharacter{1D436}{}

\begin{document}

\title{Quantum trajectories and reduced dynamics in time-correlated environments}

\author{Pietro De Checchi}
\email[Corresponding at: ]{pietro.dechecchi@phd.unipd.it}
\affiliation{Department of Mathematics, University of Padova, Via Trieste 63, Padova 35131,  Italy}
\author{Federico Gallina}
\altaffiliation{Currently at Department of Physics, University of Ottawa, Canada}
\affiliation{Department of Chemical Sciences, University of Padova, Via Marzolo 1, Padova 35131,  Italy}
\author{Barbara Fresch}
\affiliation{Department of Chemical Sciences, University of Padova, Via Marzolo 1, Padova 35131,  Italy}
\affiliation{Padua Quantum Technologies Research Center, University of Padova, via Gradenigo 6/A, Padua 35131, Italy}
\author{Giulio G. Giusteri}
\affiliation{Department of Mathematics, University of Padova, Via Trieste 63, Padova 35131,  Italy}
\affiliation{Padua Quantum Technologies Research Center, University of Padova, via Gradenigo 6/A, Padua 35131, Italy}

\date{January 21, 2026}

\begin{abstract}
    The stochastic \schr equation (SSE) provides a trajectory-level route to simulate the dynamics of open quantum systems with applications ranging from molecular processes to quantum technologies. 
    We study a colored-noise extension of the SSE based on an Ornstein–Uhlenbeck (OU) noise drive, and benchmark its ensemble-averaged dynamics against the standard white-noise SSE and against a fluctuating OU random Hamiltonian.
    When the environment exhibits a finite correlation time, averaging over pure-state trajectories yields master equations that are generally open-form and not of Lindblad type, yet remain positive by construction.
    By considering the differential of the OU process, we define an effective correlated noise, whose properties we analyze and use to formulate an SSE unraveling of its associated open-form quantum master equation. 
    We show that the averaged dissipator separates into a Lindblad contribution stemming from the white-noise component, and additional correlation terms arising from the fluctuations of the OU Hamiltonian.
    To obtain a practical closed description and physical intuition, we introduce a Redfield-inspired perturbative closure for these correlation terms, providing an effective master equation for the colored SSE. 
    For a two-level system, the resulting dynamics exhibit long-lived coherences, nontrivial stationary (including oscillatory) states, and multi-timescale relaxation, rationalized through the components of a time-dependent Redfield tensor. 
\end{abstract}

\maketitle

\section{Introduction \label{sec:Intro}}

The understanding of quantum systems interacting with fluctuating and structured environments remains a central topic in many areas of physics and chemistry, from spectroscopy to charge and energy transfer, to quantum technologies, to name a few. 

For example, in quantum computing, where coherences are exploited for computational purposes, and dissipation is typically detrimental, a detailed characterization of the environmental noise is essential to achieve so-called error mitigation and error correction. \cite{Rahman2021QuantumNoise,Butler2024OptimizingHare, Cattaneo2023QuantumComputers}
On the other hand, the interplay between coherent evolution and environment-induced dephasing can enhance the efficiency of energy transport in chromophoric aggregates of some photosynthetic organisms,\cite{Plenio2008Dephasing-assistedBiomolecules, Rebentrost2009Environment-assistedTransport} inspiring the design of artificial light-harvesting materials and systems mimicking photosynthesis.\cite{Uchiyama2018EnvironmentalTransport} 

From a theoretical point of view, an effective description of the influence of the environment on the dynamics of an open quantum system is a long-standing challenge that led to the formulation of different dynamical models and implementation schemes. In this work, we contribute to this program through the investigation of a stochastic \schr equation in which non-Markovian effects are introduced through the differential of an exponentially correlated Ornstein–Uhlenbeck process \cite{Barchielli2010StochasticNoise} rather than through white-noise increments. 

Stochastic \schr equations (SSEs) provide a trajectory-based representation of reduced dynamics that complements density-matrix approaches.\cite{Biele2012ASystems,Barchielli2009Quantum782,petruccione2002,Jacobs2006AMeasurement}
By propagating pure quantum states driven by stochastic processes and reconstructing observables through ensemble averaging, SSEs offer both conceptual transparency and numerical scalability, particularly advantageous for high-dimensional systems for which parametrization and integration of master equations become prohibitive.\cite{Coccia2018ProbingSystems,DallOsto2024StochasticSystems}
In their most commonly used form, SSEs are formulated in terms of stochastic differentials driven by Wiener processes, reflecting the assumption of short environmental memory times. Under this ``white-noise'' assumption, the ensemble-averaged dynamics is governed by a Markovian master equation of Lindblad form, ensuring complete positivity and trace preservation of the reduced density operator. 
This correspondence between diffusive SSEs driven by Wiener increments and Lindblad generators underpins the standard quantum trajectory formalism\cite{Barchielli2009Quantum782, petruccione2002} and its interpretations in terms of specific environment models or measurement protocols, from quantum jumps to continuous weak measurement back-action.\cite{petruccione2002,Jacobs2006AMeasurement,Barchielli2009Quantum782}

However, the assumption of white noise and the resulting Markovian dynamics are often inadequate in molecular sciences and emerging quantum technologies, where environments exhibit finite correlation times, structured spectral densities, and nontrivial back-action effects.
Notable frameworks were introduced to overcome the white noise formulation. 
A central line of development starts from microscopic system–boson-bath models and yields formally exact diffusive non-Markovian stochastic \schr equations 
whose driving noise inherits the bath correlation function: the paradigmatic example is non-Markovian quantum state diffusion.\cite{Diosi1997TheSystems,Diosi1998Non-MarkovianDiffusion,Stockburger2001Non-MarkovianDiffusion,Suess2014HierarchyDynamics} 
However, because of the technical difficulties encountered in solving these equations, studies reporting on applications remain scarce. 
A different strategy starts from the target time-local master equation (not necessarily in Lindblad form, e.g., the Redfield equation) and constructs the stochastic unraveling. 
In some cases, this approach can require doubling the Hilbert space and allowing negative trajectory weights to reproduce dynamics that are not completely positive by construction.\cite{Gaspard1999Non-MarkovianEquation,Gaspard1999SlippageEquation,Kleinekathofer2002StochasticClass,Kondov2003StochasticProblems,Vogt2013StochasticEnvironment}

In this work, we elaborate on another perspective, hinging on a direct generalization of the white noise version of the SSE. 
The memory term is introduced in the SSE dynamics by an explicit colored noise process, with no reference to a microscopic quantum bath. This approach has been introduced in Ref.~[\onlinecite{Barchielli2010StochasticNoise}] focusing on its interpretation in terms of a continuous-measurement process. 
More recently, the same formal setting has been elaborated with a more practical focus in Refs.~[\onlinecite{DeKeijzer2025QubitNoise,DeKeijzer2025Fidelity-enhancedControl,JansevanRensburg2025FidelityNoise}] where the origin of the colored noise is identified with the random fluctuations of qubit driving fields.
Here, we will analyze the average dynamics generated by the colored SSE through a direct comparison with the analog dynamics generated by (i) the standard white noise SSE and (ii) a random Hamiltonian defined through the related OU process. 
The formal presentation and developments are kept general, while the numerical investigation focuses on the dynamics of a two-level system: the minimal model relevant both in quantum information processing (qubit) and in modeling internal relaxation and energy transfer in molecular systems.\cite{Reimers2015AAromaticity,Giavazzi2022TheDyes} 

On the general ground of the master equation obtained by averaging the quantum trajectories generated by the colored SSE, we clearly identify two distinct contributions to dissipation: (i) a Lindblad-form term arising from the Wiener increment, in direct analogy with the standard Markovian SSE setting, and (ii) additional terms stemming from the statistical correlation between the system evolution and the colored noise used to model the environment. 
The same correlation terms also arise in the average dynamics induced by an Ornstein–Uhlenbeck (OU) random Hamiltonian. 
In general, these terms imply that the reduced evolution cannot be cast into a closed master equation; nevertheless, it remains positive by construction, since it is obtained by averaging over pure-state trajectories. 
We therefore introduce an approximate closure for these open correlation terms, based on a perturbative treatment of the system–environment interaction, as commonly done in Redfield theory.\cite{Redfield1957OnProcesses, Davies1974MarkovianEquations,Vacchini2024OpenTheory} 
This leads to the first main result of this work: a microscopically-inspired model for the correlation contributions that appear in the SSE formulation. 
The resulting link between the Redfield-like derivation and the colored SSE is particularly useful for interpreting the dynamics obtained from direct numerical simulations of the latter. 
Indeed, as a second main result, we show how the colored noise source imprints distinctive signatures on the quantum dynamics. 
In particular, for the two-level system studied here, we observe the emergence of multiple dissipation time scales and the formation of a robust coherent steady state, depending on the noise correlation time and on the symmetry properties of the system. 
Importantly, these features are rationalized within the Redfield-based closure of the open correlation terms, thereby building a transparent cause–effect intuition connecting the noise statistics to non-Markovian dynamics of the open quantum system. 
We highlight that the use of methods based on stochastic trajectories is especially relevant for quantum simulation with digital quantum computers.\cite{Lloyd1996UniversalSimulators, Hu2020ADevices,Miessen2022QuantumDynamics,Gallina2024SimulatingAlgorithm} 
To overcome the fact that quantum gates are unitary operations, one can exploit stochastic averages to implement intrinsically non-unitary dynamics of the mean density matrix.
In particular, we discuss how using constraints such as the hermiticity of the operators associated with the stochastic potential, each trajectory can be implemented as a unitary evolution of the system and thus be mapped into a quantum circuit.
Moreover, with digital quantum computers, the parallelization of the trajectories is accounted for by the intrinsic need for many runs of quantum circuits,\cite{gallina2022,Gallina2024FromDynamics}
which would be \textit{de facto} also required by other approaches based, e.g., on the dilation of the Hilbert space.\cite{Hu2020ADevices,Sweke2015UniversalSystems,Sweke2016DigitalDynamics,Schlimgen2022QuantumOperators}

The manuscript is structured as follows: in the following section, we briefly introduce the SSE method to open quantum systems, the main stochastic approach used in this work. 
We discuss the general form of the stochastic differential \schr equation without imposing the typical white noise form commonly encountered in applicative studies. 
In \Cref{sec:ModelEnvironment} we discuss the different stochastic processes and noises we will use to describe dissipation, their properties, and the rationale of the possible choices.
In \Cref{sec:QuantumDynamics} we derive the effective stochastic equations describing the evolution of the system, and we obtain the relative QME.
In \Cref{sec:RED} a time-dependent Redfield dissipator is
analytically derived in order to have a closure model for the colored noise-derived QME.
The investigation of the Redfield relaxation channels for a two-level system is presented.
The numerical results of the average evolutions and the main features of the trajectories are shown and discussed for different system and environment settings in \Cref{sec:NumericalResults}, as well as through the approximated Redfield model. 
We then summarize the main points and comment on their relevance in the concluding section.


\section{Stochastic \schr equation \label{sec:StochAppr}}

A simple and general way to formulate the SSE is with a linear homogeneous stochastic differential equation (SDE)
    \begin{equation}
    \label{eq:genericSSE}
        \diff{\psi}_t = A{\psi}_t\diff t + B{\psi}_t\diff X_t,
    \end{equation}
where the operator $A$ contains the deterministic (Hamiltonian)
evolution of the system and $B$ is the operator setting the structure of the stochastic fluctuations encoded by the \ito's differential $\diff X_t$ of a stochastic process $X_t$.
The SSE does not, in principle, ensure the preservation of the norm of the state vector during a trajectory. 
For linear SSEs, the normalization of the mean state is obtained 
ensuring the martingale property
$\|\psi(0)\|^2=\|\psi(t)\|^2=1$.
This gives the normalization condition
\begin{equation}
    \mathbb{E}\left[||\psi_t||^2\right]=1 \implies \mathbb{E}\left[\diff\big(\psi_t\daggah\psi_t\big)\right]=0
\end{equation}
which, without specifying yet the process $(X_t)_{t\geq0}$, translates to
\begin{align} 
\label{eq:conditionAB}
\begin{split}
    \bra{\psi_t}A\daggah+A\ket{\psi_t}\diff t +
    \bra{\psi_t}B\daggah+B\ket{\psi_t}\diff X_t +\\
    +\bra{\psi_t}B\daggah B\ket{\psi_t}(\diff X_t)^2 + \mathcal{O}(\diff t \diff X_t) = 0,
\end{split}
\end{align}
where the second-order term w.r.t.\ $\diff X_t$
depends on the specific process $(X_t)_{t\geq0}$ chosen. 
Equation \ref{eq:conditionAB} imposes constraints on the forms of the operators $A$ and $B$.

From the normalized states, we can then write the average density matrix, as
\begin{equation}
\label{eq:rho_as_average}
    \rho_t = \mathbb{E}\left[\ketbra{\psi_t}{\psi_t}\right],
\end{equation}
and the associated QME as
\begin{equation} \label{eq:general_associated_QME}
\begin{split}
    \diff \rho_t = &\ A \rho_t \diff t + \rho_t A^\dagger \diff t + B \mathbb{E}\left[\ketbra{\psi_t}{\psi_t} \diff X_t\right] + \\
    &+\mathbb{E}\left[\ketbra{\psi_t}{\psi_t} \diff X_t\right] B^\dagger \\
    &+ B \mathbb{E}\left[\ketbra{\psi_t}{\psi_t} (\diff X_t)^2\right] B^\dagger.
\end{split}
\end{equation}

A particular case of stochastic approaches to open systems is the use of stochastic processes as fluctuations of the Hamiltonian, instead of the use of a stochastic differential as a source. \cite{DeChecchi2026OnHamiltonians} 
The Hamiltonian becomes random and time-dependent and can be interpreted as an effective, stochastic Hamiltonian.
This approach can be considered a subclass of the SSE framework, where the coefficient of the \ito differential $\diff X_t$ in \cref{eq:genericSSE} is set to zero, and the stochasticity is due to the stochastic $A_t$ operator, defined as
\begin{equation}
    iA_t = 
    H^\mathrm{eff}_t = H + H_t^\mathrm{fluct} = 
    H + R Z_t,
\end{equation}
where $(Z_t)_{t\geq0}$ is a real-valued stochastic process, and $R$ is a Hermitian operator that preserves the Hamiltonian structure and ensures the preservation of the norm. 
Using a continuous process (i.e., the integrated version of an \ito differential) keeps the rules of ordinary calculus valid; therefore, the evolution of a single copy of the system follows the dynamics of an ordinary \schr equation
\begin{equation}
\label{eq:stoch_hamilt_SE}
    \frac{\diff \psi_t}{\diff{t}} = -i\left(H + R Z_t \right)\psi_t,
\end{equation}
and the dynamics of the average density matrix is obtained by \cref{eq:rho_as_average}. 

Notice that modeling the environment with the stochastic process $(Z_t)_{t\geq0}$ or with its underlying differential $\diff{Z_t}$ will lead to very different results, as we are dealing with different stochastic drives. Their distinctive effects on the average dynamics of the density matrix will be discussed in the following sections.  


\section{Stochastic modeling of the environment \label{sec:ModelEnvironment}}

Aiming to characterize the environment stochastically, it is important to describe the mean, variance, and covariance of the stochastic processes and noises driving the SSE. 
Many examples of different phenomenological baths can be found in the literature.\cite{Haken1973AnMotion,Caldeira1983PathMotion,Gardiner2000QuantumOptics,petruccione2002,Weiss2012QuantumEdition}
The use of $\delta$-correlated noise, i.e., white noise, is one of the most frequent, as it can be used to unravel Markovian QMEs in Lindblad form, as we shall recall.
On the other hand, considering a colored noise is one straightforward way to introduce memory effects into the open system's dynamics, computationally more efficient than other approaches that require computing memory kernels for the complete density matrix dynamics.

\subsection{White noise}

White noise $\xi_t$ is an idealization of interactions with environments with correlation times much shorter than the characteristic time of the system, which can therefore be considered memoryless.
It can be heuristically defined as the formal derivative of the Wiener process $W_t$ (Brownian motion) 
$\xi_t=\gamma\diff{W}_t/\diff{t}$.
This is a distribution, as $W_t$ is not differentiable. 
This noise has the following properties:
        (i) its mean is null, $\mathbb{E}[\xi_t]=0$,
        and
        (ii) it is $\delta$-correlated in time, $\mathbb{E}[\xi_t^{(i)}\xi_s^{(j)}]=\gamma^2\delta(t-s)$,
where 
$\gamma$ is the intensity of the noise and
$\delta(t-s)$ is the Dirac delta of the time difference. 
The definition of \textit{white} noise comes from the spectral density of the process, i.e., the Fourier transform of its autocorrelation function, which is constant (\Cref{fig:NoisesSpectralDensities}): all frequencies are contained in and equally contribute to the correlation function.

The interest, in the SSE framework, is in the increment $\diff W_t$, the \ito's differential for Brownian motion. 
This is the increment of a Gaussian and Markov process with law
$\Delta{W(t-s)}\sim \mathcal{N}(0,\sqrt{t-s})$.
In the infinitesimal limit, this can be written as $\diff W_t\sim\mathcal{N}(0,\sqrt{dt})$, 
and can then be used in \cref{eq:genericSSE} as the stochastic potential,
weighted by its intensity $\gamma$, which takes units of a square root of energy.

\subsection{Ornstein-Uhlenbeck in the stochastic Hamiltonian}

One possible generalization from white noise to correlated (colored) processes is the use of the Ornstein-Uhlenbeck (OU) process.\cite{Uhlenbeck1930OnMotion,Gillespie1996ExactIntegral}
It is often used in the stochastic Hamiltonian approach to describe finite-memory effects induced by an overdamped
environment whose relaxation time is comparable to the characteristic timescale of the system.\cite{Kubo1969AShape,Dijkstra2015CoherentVibrations,Bondarenko2020ComparisonSystems,Gallina2024FromDynamics}
The OU process is the random valued process $(X_t)_{t\geq0}$ that satisfies the following SDE:
\begin{equation}
    \diff X_t = -\theta X_t\diff t + \gamma\diff W_t, \quad\quad\quad \theta,\gamma>0,
    \label{eq:OrnsteinUhlenbeckProcess}
\end{equation}
where $\theta$ is the inverse of the correlation time of the process,
$\gamma$ is the amplitude of the random fluctuation
and $W_t$ is a Wiener process.
The explicit solution for the OU process is
\begin{equation}
    \label{eq:OrnsteinUhlenbeckProcess_Integral}
    X_t = X_0 e^{-\theta t} + \gamma\int_0^t e^{-\theta|t-s|}\diff{W_s}.
\end{equation}
Then, $(X_t)_{t\geq0}$ is a Gaussian-Markov process, endowed with the following statistical properties:

\begin{align}
    \mathbb{E}(X_t) &= X_0 e^{-\theta t} \xrightarrow{t\to\infty} 0,\\
    \text{var}(X_t) &= \frac{\gamma^2}{2\theta}\left( 1 - e^{-2\theta t} \right) \xrightarrow{t\to\infty} \frac{\gamma^2}{2\theta},
\end{align}
\begin{align}
\label{eq:cov_X}
    \begin{split}
    \text{cov}(X_t, X_{s=t-\Delta t}) = \frac{\gamma^2}{2\theta}&\left( e^{-\theta |\Delta t|} -e^{-\theta (2t+\Delta t)} \right)\\ 
    &\xrightarrow{t\to\infty} \frac{\gamma^2}{2\theta}e^{-\theta |\Delta t|},
    \end{split}
\end{align}
where the $t \to \infty$ limit represents the stationary process, i.e., when the environment is at its equilibrium. 
Therefore, under this hypothesis for the environment, it is natural to define $X_0$ as a random variable itself, in particular a Gaussian with zero mean and the same variance as the stationary process.  
Again, the correlation function of the OU process describes the correlation of the bath in which the open system is embedded.
This correlation function is related, \textit{via} Fourier transformation, 
to a Lorentzian spectral density (\Cref{fig:NoisesSpectralDensities}) of the form
\begin{equation}
    \label{eq:SD_OU-Lorentzian}
    J_X(\omega) = \int_{-\infty}^{+\infty} \text{cov}(X_t, X_0) e^{i\omega t} \diff t = \frac{\gamma^2}{\omega^2+\theta^2} .
\end{equation}

\subsection{Ornstein-Uhlenbeck in the stochastic \schr equation}

The OU process, as described so far, can be used in the stochastic Hamiltonian approach, in \cref{eq:stoch_hamilt_SE}.
However, we note that its SDE, \cref{eq:OrnsteinUhlenbeckProcess}, can also be used as the stochastic potential to drive our SSE, in \cref{eq:genericSSE}.\cite{DeChecchi2026OnHamiltonians} 
In this way, we have to deal with a new noise, which is the derivative of the OU process, 
in analogy with what was done previously for white noise and the Wiener process.

We can give a formal definition of this noise as the time derivative of $(X_t)$,
in the distribution sense,
dividing by $\diff{t}$ and expressing the integral form of $(X_t)$:
\begin{equation}
    \label{eq:Yprocess_noise}
     \Upsilon_t = \dot{X}_t 
         =  -\theta X_0 e^{-\theta t} - \theta\gamma\int_0^te^{-\theta (t-s)}\diff W_s + \gamma\xi_t. 
\end{equation} 

The mean of $\Upsilon_t$ is again zero, as the stochastic terms average to zero, and we consider the noise to be constructed over the stationary OU process. We have
\begin{equation}
    \mathbb{E}[\Upsilon_t] = -\theta e^{-\theta t} \mathbb{E}[ X_0]=0   
\end{equation}
for  $X_0 \sim \mathcal{N}(0, \text{var}(X_\infty)),\; t\to\infty$,
and the covariance of the noise (see \Cref{app:CovY_computing} in the supplementary material) is
\begin{equation}
\label{eq:cov_Y}
    \text{cov}(\Upsilon_t,\Upsilon_s) =  -\frac{\gamma^2 \theta}{2} e^{-\theta |t-s|} + \gamma^2\delta(t-s).
\end{equation}

It is interesting to note that the correlation function includes both the characteristic $\delta$-function of the white noise and the exponential decay of the OU process.
The spectral density of this noise is then a constant minus the rescaled Lorentzian
\begin{equation}
    \label{eq:SD_OU-NOISE}
    J_\Upsilon(\omega) = \int_{-\infty}^{+\infty} \text{cov}(\Upsilon_t) e^{i\omega t} \diff t 
    = \gamma^2 \frac{\omega^2}{\omega^2+\theta^2} ,
\end{equation}
which converges asymptotically to 1 at infinity and collapses to 0 at low frequencies (\Cref{fig:NoisesSpectralDensities}).
This particular form of the spectral density has consequences on the system dynamics, as we will discuss in the following sections, both analytically and with numerical implementations.

\begin{figure}
    \centering
    \includegraphics[width=0.8\linewidth]{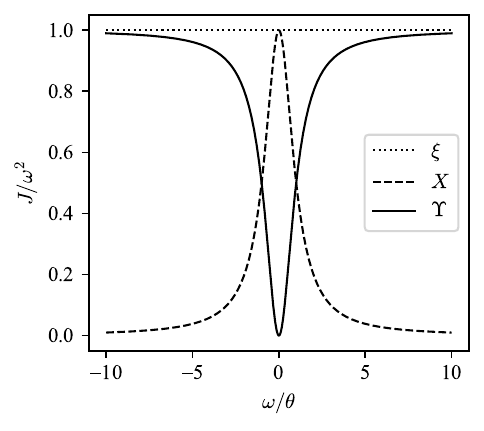}
    \caption{Spectral densities of white noise ($\xi$, dotted line), OU process ($X$, dashed line) and $\Upsilon_\mathrm{OU}$-noise ($\Upsilon$, solid line).
    }
    \label{fig:NoisesSpectralDensities}
\end{figure}

\section{Quantum dynamics \label{sec:QuantumDynamics} }

\subsection{Lindblad-form master equations from white noise \label{sec:LindbladfromWhite}}

We now consider the generic SSE in \cref{eq:genericSSE} subject to white noise, which reads
\begin{equation}
    \diff\psi_t = A\psi_t\diff t + \gamma B\psi_t\diff W_t.
\end{equation}
This form of the SSE needs to be normalized according to \cref{eq:conditionAB}.
The noise components average to zero due to the properties of $\diff W_t$,
and, to preserve the martingale property, we have to set equal to zero the sum of the operators that multiply $\diff t$, arriving at
\begin{equation}
\label{eq:normal_A}
    A\daggah + A + \gamma^2 B\daggah B = 0.
\end{equation}

This leads to the identification of the operator of the deterministic part as the sum of the Hamiltonian term and a normalization term, namely 
\begin{equation}
\label{eq:Aoperatorwithcorrection}
    A = \frac{i}{\hbar} H -\frac{1}{2}\gamma^2 B\daggah B,    
\end{equation}
where $(B\daggah B)$ is always Hermitian. To keep consistent units, we set $B=L/\sqrt{\hbar}$.
Then, the normalized SSE becomes
\begin{equation}
\label{eq:normalized_white_sse}
    \diff\psi_t = \left(-\frac{i}{\hbar}H -\frac{1}{2\hbar}\gamma^2 L\daggah L \right)\psi_t \diff t 
    + \frac{\gamma}{\sqrt{\hbar}} L\psi_t\diff W_t,
\end{equation}
with the associated QME (see eq. \ref{eq:general_associated_QME}) given by
\begin{equation}
\label{eq:lindblad}
    \difft \rho_t = 
    -\frac{i}{\hbar}\left[H,\rho_t\right] 
    + \frac{\gamma^2}{\hbar} \left( L\daggah\rho_t L - \frac{1}{2}\left\{L\daggah L, \rho_t\right\} \right),
\end{equation}
where $\{\cdot,\cdot\}$ is the anti-commutator. 
In this QME, we recognize a Lindblad form with operator $L$ as the jump operator
and $\gamma^2/\hbar=\Gamma/{\hbar}$
as the dissipation rate. 
The approach is easily extended to multiple dissipation channels $(L_k,\Gamma_k)$ by introducing multidimensional noises and a vector of intensities $\gamma_k$, associated with each operator $L_k$.
Since we can freely choose the $L_k$ operators,
we can recover any Lindblad-form QME from an SSE.\cite{Biele2012ASystems,Donvil2022QuantumEquations} 
For instance, the absence of constraints on $L_k$ allows us to consider non-Hermitian operators, leading to stationary probability distributions that are not uniformly distributed on all Hamiltonian eigenstates.\cite{DeChecchi2026OnHamiltonians}

\subsection{Correlated QME from Ornstein-Uhlenbeck-driven SSE \label{sec:OU_SSE}}

To introduce memory effects in the dynamics, we can use a different noise source.
Starting from the general non-normalized SSE, \cref{eq:genericSSE},
we substitute the Ornstein-Uhlenbeck SDE, \cref{eq:OrnsteinUhlenbeckProcess},
as the noise source.
Then, the non-normalized SSE is now
\begin{equation}
    \diff\psi_t = (A-\theta X_t B)\psi_t\diff t + \gamma B\psi\diff W_t.
    \label{eq:nonNorm-OU-SSE}
\end{equation}

In the normalization, the martingale condition is satisfied when
\begin{equation}
\label{eq:normalization_OUSSE}
    \langle\psi| \left[ A\daggah+A -\theta X_t (B\daggah+B) + \gamma^2B\daggah B \right]| \psi\rangle dt = 0.
\end{equation}

Compared to the white noise case, to keep operators $A$ and $B$ deterministic and time-independent, 
we need to impose further constraints on the $B$ operator, which must be an anti-Hermitian operator of the form 
$B = -iR/\sqrt{\hbar}$ with  $R\daggah=R$ (see \Cref{app:proofAntiHermit} in the supplementary material for more details). 
Substituting $B$ in \cref{eq:normalization_OUSSE},
we obtain the same normalization term in $A$ as in the white-noise case, \cref{eq:Aoperatorwithcorrection}, since the sum of anti-Hermitian operators cancels out. 
The linear normalized OU-driven SSE is then
\begin{equation}
    \label{eq:OUSSE_linear}
    \begin{split}
     \diff\psi_t = &\left( -\frac{i}{\hbar}H 
     -\frac{1}{2\hbar}\gamma^2 R\daggah R \right)\psi_t \diff t\\
     &+ \frac{i}{\sqrt{\hbar}}\theta X_t R \psi_t  \diff t\\
     &-  \frac{i}{\sqrt{\hbar}}\gamma R\psi_t\diff W_t,
    \end{split}
\end{equation}
as obtained, for instance, in ref.~[\onlinecite{Barchielli2010StochasticNoise}]. 
In this, we clearly see the addition of a term proportional to the colored noise process $X_t$.
Taking the derivative of the outer product of the wavefunction, we first obtain the QSME associated with a single trajectory (denoted as $\rho^\mathrm{trj}_t$) as
\begin{equation}
    \label{eq:QSME-OU-purestate}
\begin{split}
    \diff(\rho_t^\mathrm{trj}) = 
    &-\frac{i}{\hbar}[H,\rho_t^\mathrm{trj}]\diff t
    +\frac{i}{\sqrt{\hbar}}[\theta X_tR,\rho_t^\mathrm{trj}]\diff t\\
    &-\frac{\gamma^2}{2\hbar}[{R},{[{R},{\rho_t^\mathrm{trj}}]}]\diff t
    -\frac{i}{\sqrt{\hbar}}\gamma[{R},{\rho_t^\mathrm{trj}}]\diff W_t.
\end{split}
\end{equation}

By averaging over the trajectories, the mean density matrix dynamics is governed by
\begin{equation}
    \label{eq:nM-QME_OUdriven}
    \begin{split}
    \difft \rho_t = &-\frac{i}{\hbar}[{H},{\rho_t}]
    + \frac{\gamma^2}{\hbar} \left( R\rho_t R\daggah - \frac{1}{2}\{R\daggah R,\rho_t\} \right)\\
    &+\frac{i\theta}{\sqrt{\hbar}}[{R},{\mathbb{E}(X_t \rho^\mathrm{trj}_t)}].
    \end{split}
\end{equation}
The right-hand side presents itself with three distinct terms.
The first one accounts for the contribution of the system Hamiltonian. 
The second one, proportional to $\gamma^2$, is a Lindblad-type dissipator.
The last term, proportional to $\theta$, features the correlation ${\mathbb{E}(X_t \rho^\mathrm{trj}_t)}$ between the system evolution and the colored noise.

Due to this cross-correlation between the single realization of the OU process and each trajectory, 
the QME in \cref{eq:nM-QME_OUdriven}
is not in a closed form in terms of the mean density matrix $\rho_t$, and we have to rely on stochastic unraveling for its numerical solution. 
This affects the long-time dynamics in a way that is not immediately evident or easily predictable.
In \cref{sec:RED} below, we derive an approximate closure model that helps the understanding of the role of the correlation term.

A striking example of the importance of the correlation term in \cref{eq:nM-QME_OUdriven} can be seen by considering its stationary state.  Because the operators $R$ are constrained to be Hermitian when introducing colored noise (see \Cref{app:proofAntiHermit} in the supplementary material), by considering only the Lindblad dissipator we would expect to reach a distribution with equally populated eigenstates.
We will show in the following that this prediction is not always true when the open cross-correlation term is present.


\subsubsection{Example: element-wise EOM from the correlated QME \label{sec:EOMs}}

To understand the direct consequences of substituting white noise with a colored noise source in the SSE, let us look at the simplest case of a two-dimensional Hilbert space spanned by the basis vectors $\{\ket{0},\ket{1}\}$.
We consider two resonant states and set to zero the deterministic part of the interaction Hamiltonian.
We choose $R=\sigma_x$ ($\sigma_x$, $\sigma_y$, $\sigma_z$ are Pauli matrices)
to obtain a correlated version of the Pauli master equation with symmetric decay rates.
The resulting QME reads
\begin{equation}
\label{eq:pauli-like_nM-QME}
    \difft \rho_t = \frac{\Gamma}{\hbar} \left( \sigma_x \rho_t \sigma_x - \rho_t \right) 
    + \frac{i\theta}{\sqrt{\hbar}}[{\sigma_x},{\mathbb{E}(X_t \rho^\mathrm{trj}_t)}],
\end{equation}
where $\Gamma=\gamma^2$ would be the relaxation rate of the channel in the absence of memory effects. 

The variables of the system are the four elements of the density matrix. 
By normalization, the populations are such that 
$\rho_{00}+\rho_{11}=1$ and by hermiticity the coherences satisfy $\rho_{01}=\rho_{10}^*$.
So we can describe the system simply by computing the time derivatives of two variables $\rho_{00}$ and $\rho_{10}$.
The dynamics of the population is
\begin{equation}
\label{eq:2sites_population_EOM_nMQME}
    \dot{\rho}_{00} = \frac{\Gamma}{\hbar} ( 1 - 2\rho_{00}) 
                    - \frac{2\theta}{\sqrt{\hbar}} \mathbb{E}\left[ X_t \Im( \rho^\mathrm{trj}_{10} ) \right]
\end{equation}
while the coherence evolves with
\begin{equation}
\label{eq:2sites_coherences_EOM_nMQME}
    \dot{\rho}_{10} = -\frac{i2\Gamma}{\hbar} \Im(\rho_{10}) 
                        + \frac{i2\theta}{\sqrt{\hbar}} \mathbb{E}\left[  X_t \rho^\mathrm{trj}_{00} \right] 
                        - i\theta\mathbb{E}(X_t)
\end{equation}
where we know the mean of the OU process,
$\mathbb{E}(X_t)= X_0 e^{-\theta|t-t_0|}$, 
conditioned on the initial value $X_{t_0}$.

In this simple case, the effect of the cross-correlation term
is to  
mix \textit{trajectory-wise}  the equations of motion of populations and coherences. 
In particular, populations are affected by the correlation of coherences with the OU process.
This is an effective correlation of the system coherences with the environment. 
These correlations are usually neglected in the microscopic derivation of quantum master equations (see \Cref{sec:RED}) 
when assuming the factorized form of the total density matrix in terms of system and bath density matrices at all times (Born approximation). 
The last term in \cref{eq:pauli-like_nM-QME} is also responsible for the possible deviation of the asymptotic value of populations from equipartition and for non-vanishing steady-state coherences.  

\subsection{Stochastic Hamiltonian with Ornstein-Uhlenbeck fluctuations}

Within the framework of effective stochastic Hamiltonian,\cite{DeChecchi2026OnHamiltonians} 
the Ornstein-Uhlenbeck process $(X_t)_{t\geq0}$
is often used as a random potential (not its noisy increments).
As mentioned in \cref{sec:ModelEnvironment}, this describes the interaction with an environment of overdamped oscillators.
The effective Hamiltonian can be written as
\begin{equation}
    H_\mathrm{eff} = H + \sqrt{\hbar} R \theta X_t ,
\end{equation}
where the parameter $\theta$ is necessary to keep consistent units, and the dynamics of the system becomes
\begin{equation}
    \label{eq:StochHamilt_notSDE_X}
    \diff \psi_t = -\frac{i}{\hbar}\left(H + \theta\sqrt{\hbar} R X_t \right)\psi_t \diff t.
\end{equation}

The average Liouville equation associated with this dynamics reads
\begin{equation}
\label{eq:qme_Xou}
    \difft \rho_t =  -\frac{i}{\hbar} [H,\rho_t] -\frac{i\theta}{\sqrt{\hbar}} [R,\mathbb{E}(X_t \rho^\mathrm{trj}_t) ].
\end{equation}
As can be noted, here, we do not find any Lindblad-form dissipator. The only term in addition to the Hamiltonian contribution is the correlation term already observed in \cref{eq:nM-QME_OUdriven}, albeit with opposite sign.

We can now identify a correspondence between components of the environment correlation function and terms of the dissipator in the mean dynamics.
The ensemble average of the SSE with a white-noise potential, which is strictly $\delta$-correlated, yields Lindblad-type dissipators, \cref{eq:lindblad}. In contrast, employing the exponentially correlated $X_\mathrm{OU}$ process results in a dissipator containing a single term that explicitly reflects the system–environment correlation, \cref{eq:qme_Xou}.
Finally, the QME obtained by averaging the SSE driven by the colored noise $\Upsilon_\mathrm{OU}$, whose correlation function includes both delta and exponential correlations, \cref{eq:cov_Y}, features a Lindblad dissipator together with an additional term that captures environment memory effects, \cref{eq:nM-QME_OUdriven}.

\section{Modeling with the Redfield microscopic approach \label{sec:RED}}

In this section, we explicitly derive a closed model master equation for the systems under investigation, interacting with colored environments.
To do so, we elaborate on the relation between the stochastic approach and the microscopic derivation of the quantum master equation in the form developed by Redfield.\cite{Redfield1957OnProcesses, Davies1974MarkovianEquations,Vacchini2024OpenTheory}

Let us start by summarizing the main points of the Redfield derivation to identify explicitly the underlying assumptions and approximations (see, for instance, Refs. [\onlinecite{petruccione2002,Rivas2011OpenIntroduction,Vacchini2024OpenTheory}]).
The starting point is the microscopic definition of the global Hamiltonian in terms of the system, bath, and interaction  Hamiltonians as
\begin{equation}
    H = H_S\Tensor\Id_B + \Id_S\Tensor H_B + H_I,
    \label{eq:totalHamiltonian_general}
\end{equation}
where $H_S,H_B$ describe the two isolated subsystems and the interaction Hamiltonian $H_I\in\hilbert$ is defined as
\begin{equation}
    H_I = \sum_\alpha S_\alpha\Tensor B_\alpha.
    \label{eq:interactionHamiltonianFacorization}
\end{equation}

In the interaction picture,
indicated by a tilde over the operators,
we focus only on the evolution due to the interaction Hamiltonian, now time-dependent.
We derive the system evolution by writing the Liouville equation, integrating it, and inserting the integrated density matrix into the initial equation. Then, we trace out the bath from the resulting integro-differential equation, a recasting that preserves the exact dynamic description.
Then, approximations are needed to obtain a closed and solvable equation.
    First, we assume that the initial density matrix is factorized, and the bath is considered at thermal equilibrium,
        $\rho(0) = \rho_S(0)\Tensor\rho_B^{eq}$.
    Then, considering  
    a perturbative interaction of the environment and the system,  we extend the factorization at all times $t$ by setting
        \begin{equation}
        \label{eq:factorizeddensitymatrix}
        \rho(t) \approx \rho_S(t)\Tensor\rho_B^{eq}    
        \end{equation}
    and neglecting higher-order terms in the system-bath correlation, which is known as the Born approximation.

Using \cref{eq:interactionHamiltonianFacorization} and \cref{eq:factorizeddensitymatrix}, we obtain a factorization of the system and environment degrees of freedom, where the dynamics of the environment enters through its two-time correlation function:
\begin{equation}
    C_{\alpha\beta}(t,t') = \mathrm{Tr}\left[B_\alpha(t)B_\beta(t')\rho_B^{eq} \right]. 
    \label{eq:corr_bath}
\end{equation}

The bath correlation function allows us to bridge the microscopic derivation with the stochastic approach, as the relevant bath dynamics is captured by the two-time correlation function in both cases.

Using the cyclic property of the trace and the definition \cref{eq:corr_bath}, we can now write
\begin{equation}
    \begin{split}
    \difft \tilde{\rho}_S(t) = 
    - \frac{1}{\hbar^2}
    &\sum_{\alpha\beta}
                \int_0^t   
                \Big(
                C_{\alpha\beta}(t,t')[\tilde{S}_\alpha(t), \tilde{S}_\beta(t')\tilde{\rho}_S(t')]\\
               &+ C_{\alpha\beta}(t',t)[\tilde{\rho}_S(t')\tilde{S}_\beta(t'), \tilde{S}_\alpha(t)] \Big)  \diff t'
    \end{split}
    \label{eq:nonmarkovQMEgeneral}
\end{equation}
which is a non-Markovian QME, as it depends on $\rho_S$ at previous times.

Let us now consider the simple two-level system studied above, driven by the $\Upsilon_\mathrm{OU}$ noise.
We consider only one bath channel, so we can drop the sum over the bath operators
and identify the operators acting on the system by setting $S=\sqrt{\hbar}R$.
From now on, for the sake of clarity, we denote the system density matrix as $\rho$ instead of $\rho_S$. 
By construction, the density matrix of the bath commutes with its Hamiltonian, hence the correlation function depends only on the time difference $(t-t')$,
and it corresponds to what we computed in \cref{eq:cov_Y}.
We then identify $C_{\alpha\beta}(t,t')=C(\Upsilon_t,\Upsilon_{t'})=C(t-t')$.
Moreover, since the stochastic process is classical, we can recognize that the additional symmetry $C(t-t')=C(t'-t)$ holds. 
In this setting, 
\cref{eq:nonmarkovQMEgeneral} takes the form
\begin{equation}
\label{eq:nonmarkovQME_SysSetting}
    \difft \tilde{\rho}(t) = 
                -\frac{1}{\hbar} \int_0^t   
                C(t-t')
                \left([\tilde{R}(t),
                \tilde{R}(t')\tilde{\rho}(t')]
                + \text{h.c.}\right) \diff t'.
\end{equation}

To obtain an evolution equation which is local in time, meaning that it depends on the state of the system at time $t$, a further approximation is needed. The physical assumption is that the evolution of the state is slower than the evolution of the environment, so that 
$\rho(t')\approx\rho(t)$.

We can then change the integration variable to the time difference $\tau = t-t'$, thus obtaining the \textit{Redfield equation with time-dependent coefficients}
\begin{equation}
\label{eq:Redfield_TimeDepentCoeff_QME_SysSetting}
    \difft \tilde{\rho}(t) = 
                -\frac{1}{\hbar} \int_0^t   
                C(\tau)
                \left(
                [\tilde{R}(t),
                \tilde{R}(t-\tau)\tilde{\rho}(t)]
                + \text{h.c.}
                \right) \diff \tau.
\end{equation}
This is a time-local QME, yet with a component still depending on the dynamical history of the bath. 
To recover what is typically intended as the Redfield equation, the upper limit of the time integral should be extended to infinity, a full Markovian approximation that we do not invoke here.

We can now express the operator $R$ on $\{\ket{a},\ket{b}\}$ eigenbasis of $H_S$ 
and make explicit the time dependence of the coupling operators. The action of $e^{ iH_S\tau}$ in this basis is easily evaluated. For instance, we have
\begin{equation}
    \label{}
    \begin{split}
    \ketbra{a}{a}\tilde{R}(t)\ketbra{b}{b}
    &=\ketbra{a}{a} e^{\frac{i}{\hbar} H_S t} R e^{-\frac{i}{\hbar}H_S t} \ketbra{b}{b}\\
    &= e^{i (\omega_a-\omega_b) t} \ketbra{a}{a} R \ketbra{b}{b},
    \end{split}
\end{equation}
where ${E_a}/{\hbar}=\omega_a$ and $\omega_{ba}=\omega_b-\omega_a$. 
This structure suggests the decomposition of the coupling operator according to the energy gaps of the system, that is
\begin{subequations}
\begin{align}
    \label{eq:RinEigenB}
    \tilde{R}(t) &=\sum_{
    \omega=0,\omega_{ab},\omega_{ba}
    } R(\omega) e^{-i\omega t}\\
        &= 2R(0) 
        + e^{-i\omega_{ba} t} R(\omega_{ba}) 
        + e^{-i\omega_{ab} t} R(\omega_{ab}),
\end{align}
\end{subequations}
where the tilde indicates the operators in the interaction picture, its absence the \schr picture, and the dependence of $R(\omega)$ indicates that the operator is that of the $\omega$ energy gap of the system.  
With these definitions,  \cref{eq:Redfield_TimeDepentCoeff_QME_SysSetting} can be rewritten as
\begin{equation}
    \label{eq:RedfieldFrequenx}
    \begin{split}
    \difft \tilde{\rho}(t) = 
    -\frac{1}{\hbar} \sum_{\omega,\omega'} 
    \bigg(
    &\int_0^t 
    C(\tau) e^{i\omega'\tau} \diff\tau\bigg) e^{i(\omega-\omega')t}\\
    &\times
    \left(\left[ R\daggah(\omega) , R(\omega')\tilde{\rho}(t)  \right] +\text{h.c.} \right) ,
    \end{split}
\end{equation}
where the integral defines a half-sided finite-time Fourier transformation.
Upon explicit integration, 
we obtain time- and frequency-dependent coefficients
\begin{equation}
    \label{eq:TimeFreqCoeff}
    \Gamma(\omega',t)
    = \int_0^t C(\tau) e^{i\omega'\tau} \diff\tau
    =\frac{\gamma^2}{2}\left(1-\theta\frac{e^{(i\omega'-\theta)t}-1}{i\omega'-\theta}\right),
\end{equation}
using the analytical form for the covariance of the $\Upsilon_\mathrm{OU}$ noise.
Rotating back to the \schr picture, we finally obtain 
\begin{align}
\label{eq:RedfieldFreqFINAL}
    \difft \rho(t) = 
    &-\frac{i}{\hbar}[H_S,\rho(t)] \\
    &-\frac{1}{\hbar} \sum_{\omega,\omega'}\Gamma(\omega',t)\left(\left[ R\daggah(\omega) , R(\omega')\rho(t)  \right] +\text{h.c.} \right).\nonumber
\end{align}

\subsection{Closure model for colored environments \label{sec:ClosureModel}}

We are now in the position of formulating an approximation of the cross-correlation term arising from the colored SSE unraveling.
Comparing the two QMEs, 
from the SSE \cref{eq:nM-QME_OUdriven} and the Redfield \cref{eq:RedfieldFreqFINAL}, the following relation between the dissipators emerges
\begin{subequations}
    \begin{align} 
    \label{eq:dissipators_relation_ssered}
    &\empty \frac{\gamma^2}{\hbar} \left( R\rho(t) R\daggah - \frac{1}{2}\left\{R\daggah R,\rho(t)\right\} \right)
    +\frac{i\theta}{\sqrt{\hbar}} [{R},{\mathbb{E}(X \rho^\mathrm{trj}(t))}]\approx\\
    &\approx -\frac{1}{\hbar}
    \sum_{\omega,\omega'}\Gamma(\omega',t)\left(\left[ R\daggah(\omega) , R(\omega')\rho(t)  \right] +\text{h.c.} \right).\label{eq:toapprox}
    \end{align}
\end{subequations}

By inserting the explicit form of the coefficients from \cref{eq:TimeFreqCoeff}, \cref{eq:toapprox} becomes
\begin{widetext}
\begin{subequations}
    \begin{align} 
    & \sum_{\omega,\omega'}
    \frac{\gamma^2}{2\hbar}\left(1-\theta\frac{e^{(i\omega'-\theta)t}-1}{i\omega'-\theta}\right)
    \left( 2 R(\omega)\rho(t) R(\omega')\daggah - 
    \{R(\omega)\daggah R(\omega'),\rho(t)\} 
    \right)=\\
    &= \frac{\gamma^2}{2\hbar}\sum_{\omega,\omega'}
    \Big( 2 R(\omega)\rho(t) R(\omega')\daggah - 
    \{R(\omega)\daggah R(\omega'),\rho(t)\} 
    \Big) \label{eq:firstterm}\\
    &\quad\quad+\frac{\gamma^2}{2\hbar}
    \sum_{\omega,\omega'}\theta\frac{1-e^{(i\omega'-\theta)t}}{i\omega'-\theta}
    \Big( 2 R(\omega)\rho(t) R(\omega')\daggah - 
    \{R(\omega)\daggah R(\omega'),\rho(t)\} 
    \Big).
    \label{eq:secondterm}
    \end{align}
\end{subequations}

The first term, \cref{eq:firstterm}, is just the Lindblad dissipator where the jump operators are expressed in the basis of Hamiltonian eigenstates.
By rotating back to the original frame, such Lindbladian terms are equivalent to those emerging from the SSE,
first term in \cref{eq:dissipators_relation_ssered}. 
Therefore, the approximate equivalence becomes
\begin{equation}
    \frac{i\theta}{\sqrt{\hbar}} [{R},{\mathbb{E}(X \rho^\mathrm{trj}(t))}]
    \approx
    \frac{\gamma^2}{2\hbar} \sum_{\omega,\omega'}\theta\frac{1-e^{(i\omega'-\theta)t}}{i\omega'-\theta}
    \Big( 2 R(\omega)\rho(t) R(\omega')\daggah - \{R(\omega)\daggah R(\omega'),\rho(t)\} 
    \Big).
\end{equation}
\end{widetext}
This gives the approximate expression of the cross-correlation term as a tensorial function, based on a rotated and operator-wise weighted generator, linear in the system density matrix $\rho(t)$,
effectively producing a closure model for the correlation term in \cref{eq:nM-QME_OUdriven}.
This term (with opposite sign) is the same appearing in the $X_\mathrm{OU}$-driven QME,  
allowing for a closure model of \cref{eq:qme_Xou} as well.

These Redfield equations provide an alternative viewpoint on the correlation terms derived in the previous sections and allow us to build an intuition on how they affect the system dynamics. 
In the following, we will discuss the explicit form of the time-dependent Redfield dissipator for the prototypical example of a two-level system.
We anticipate that the analysis provides the rationale to understand some unexpected features of the dynamics obtained by the numerical solution of the colored SSE presented in   \cref{sec:NumericalResults}.

\subsection{Redfield relaxation channels for the two level system driven by $\Upsilon_\text{OU}$-noise\label{sec:RedfieldModelEvolution}}

The microscopic derivation of the quantum master equation provides an explicit link between the rates of different relaxation processes and the environment dynamics as described by its correlation function. Notably, the available relaxation channels are determined by the structure of the coupling, i.e., by the system operators $R$ coupling with the bath, see \cref{eq:RedfieldFreqFINAL}. 
When $R$ commutes with the system Hamiltonian $H_S$, there is no energy exchange between the system and the environment, and the only irreversible channel is decoherence of the energy eigenbasis. The pure decoherence model corresponds to a coupling operator $R$ with components only for $\omega=0$ in \cref{eq:RinEigenB}.

Conversely, when the coupling operator does not commute with $H_S$, it presents components at different frequencies, triggering transitions between eigenstates of different energies with rates that are determined by the environment through $\Gamma(\omega,t)$ defined in \cref{eq:TimeFreqCoeff}.     

To illustrate the implication of this relaxation structure applied to the specific case of the frequency- and time-dependent transition rates associated with the
$\Upsilon_\mathrm{OU}$-noise, 
\cref{eq:TimeFreqCoeff}, let us consider again a two-level system with the generic Hamiltonian
\begin{equation}
\label{eq:GenericSystHamiltonian}
    H_{\varepsilon,\Omega} = -\frac{1}{2}\varepsilon\sigma_z + \Omega\sigma_x.
\end{equation}

We refer to the eigenstates of $\sigma_z$ as the observable basis, $\{\ket{0},\ket{1}\}$, for example the computational states of a qubit or two local sites in a transfer problem, such as exciton or charge transfer. In the following, we will assume the latter perspective. 
Therefore, we will denote different systems by setting the parameters $\varepsilon$ and $\Omega$ as follows:  two states that are energy-degenerate and not interacting ($H_{0,0}$), coupled degenerate levels ($H_{0,\Omega}$), a two-level system which is not interacting through Hamiltonian dynamics ($H_{\varepsilon,0}$), or the most general case of two coupled states at different energies ($H_{\varepsilon,\Omega}$). 

Evaluating the system operators and the time-dependent rates for each frequency leads to different results about active relaxation channels and consequent stationary state, depending on the system Hamiltonian and the noise operator (\Cref{fig:timeCoeffiecientsRedfield}).
Let us consider two degenerate states coupled by the Hamiltonian, $H_{0,\Omega}$, and the noise operator as the Pauli $\sigma_x$ operator.
This is a pure decoherence setting where the coupling to the noise commutes with the Hamiltonian of the system.
The frequencies of the system obtained by diagonalizing the Hamiltonian are $(-2\Omega,2\Omega)$, but the noise coupling operator, $R$, does not connect different eigenstates, i.e. $R_{k,l\neq k}=0$, 
nullifying the contribution of the coefficients $\Gamma(\omega\neq0)$.
On the other hand, the zero-frequency component of $R$ is not zero
and the contribution of the coefficient $\Gamma(\omega=0)$ is real and exponentially decaying in time (solid line in \cref{fig:timeCoeffiecientsRedfield}). 
Physically, this describes a decoherence channel which is active only at early times while the decoherence rate vanishes at longer times. 
This correctly predicts the peculiar stationary state obtained numerically from the SSE implementation (see next section), where we observe only partial decoherence of the initial state. 
The result is an asymptotic mixed state that can maintain some coherent behavior, depending on the timescale of the system and the intensity of the noise.

\begin{figure}
    \centering
    \includegraphics[width=0.85\linewidth]{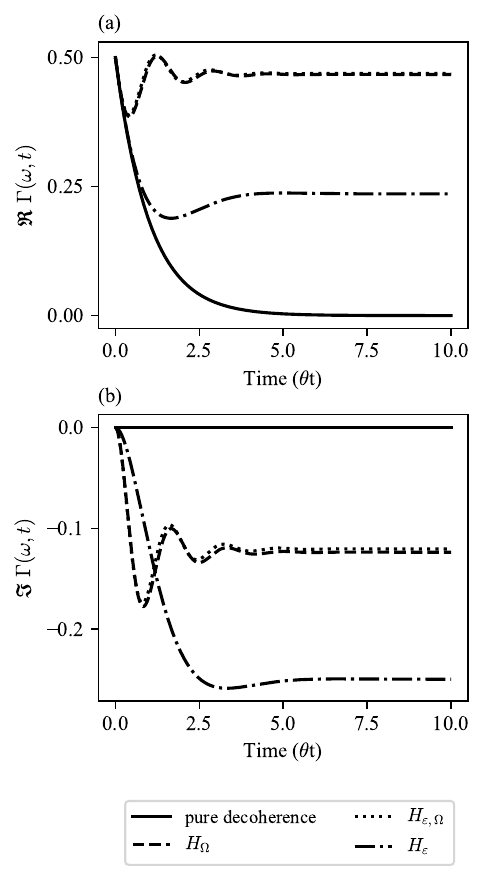}
    \caption{
    Time-dependent coefficients of the Redfield equation computed for the different Hamiltonians, with $\gamma=0.45\sqrt{\hbar\theta}$ and $\theta=1$. 
    The real part is on panel (a) and the imaginary part on panel (b).  
    Negative frequency coefficients are not shown as they are the complex conjugate of the positive ones. 
    The null frequency coefficient contribution is equal for all the Hamiltonians and labeled ``pure decoherence''.
    The characteristic frequencies of the different Hamiltonians are in different line styles as depicted in the legend.
    }
    \label{fig:timeCoeffiecientsRedfield}
\end{figure}

In the case of a non-coupled and non-degenerate system, the eigenstates of the Hamiltonian $H_{\varepsilon,0}$ correspond to the observable basis; therefore, the $\sigma_x$ coupling to the noise only induces transitions between the eigenstates.
Pure decoherence due to the zero-frequency term is null, whilst the components at non-zero frequency matter. 
These coefficients have the same value as the zero-frequency coefficient at the initial time. However, they are complex and after a few oscillations, they converge to a non-zero absolute value, lower than the initial one.
Because these coefficients weigh the upward and downward transitions equally, we should expect a stationary state reflecting the equipartition of population between the states. 

In the case of the generic Hamiltonian $H_{\varepsilon,\Omega}$, the $\sigma_x$ noise operator has both a pure dephasing component and finite frequency components inducing transitions.
Specifically, if  $U$ is the unitary transformation diagonalizing the Hamiltonian, the weight of the pure dephasing and finite frequency contributions are given by the diagonal and off-diagonal elements of the transformed coupling operator, namely
\begin{equation}
    R \mapsto U\sigma_x U\daggah
    =-\begin{pmatrix}
        \frac{\varepsilon + \sqrt{\varepsilon^2 + 4 \Omega^2}}{\Omega} & \frac{\varepsilon}{\Omega} \\
        \frac{\varepsilon}{\Omega}  & \frac{\varepsilon + \sqrt{\varepsilon^2 + 4 \Omega^2}}{\Omega}
    \end{pmatrix}.
\end{equation}

The two different main time scales for the dissipation can then be clearly seen and estimated. 
In the short time, the system decay is due to the effects of both the zero-frequency and the $2\omega'$ coefficients.
For a longer time, the decay is due only to the coefficients at non-zero frequencies. When $\varepsilon<\Omega$, the resulting relaxation can be much slower than the oscillation frequency of the system, resulting in long-lived coherent transfer between the two sites.

In the next section, we will see how the above arguments based on a Redfield-like form are realized by the dynamics resulting from the numerical solution of the SSE with colored noise. However, it is important to remark a significant difference between the two methods.
The Redfield derivation returns a model for the QME obtained by the stochastic approaches, not the same QME but
a closure \textit{model}.
Since the model QME forms are not Lindblad, we cannot guarantee \textit{a priori} the positivity of the maps. 
In contrast, the dynamics obtained by the average of stochastic trajectories are always positive by construction.

\section{Numerical results \label{sec:NumericalResults}}

In this section, we analyze the colored-noise SSE, \cref{eq:nonNorm-OU-SSE}, and its closure model through the numerical results obtained for the evolution of the two-site model 
introduced in \cref{sec:RedfieldModelEvolution}. 
We refer to the Hamiltonian in \cref{eq:GenericSystHamiltonian} and the eigenstates of $\sigma_z$ are the observable basis vectors, $\{\ket{0},\ket{1}\}$ that we interpret as different sites in a transfer problem. 
The system Hamiltonian parameters are kept constant throughout the numerical analysis,
setting $\varepsilon=\hbar\theta$ and $\Omega=2\hbar\theta$. 
The initial state of the system is set to be a pure state localized on $\ket{1}$, meaning 
$\ket{\psi_0^{(n)}}=\ket{1}$ for each trajectory $n$.

The results are obtained by averaging over $2 \cdot 10^4$ to $10^5$ trajectories realization for each dynamics, 
computed using the Euler–Maruyama method,\cite{Bally1995TheCalculus} the extension to SDE of the Euler finite difference method,
with time increments such as to ensure convergence of the noise realizations ($\theta\Delta t=10^{-5}$).
 
We will discuss specific cases, corresponding to different scenarios of how the noise affects the systems: it can either modulate the energy of the local sites (diagonal noise), or act on the coupling between the two sites. 

The effects of different stochastic sources will be compared, the 
intensity of the white noise component is set at $\gamma=0.7\sqrt{\hbar\theta}$
(the effect of changing it, from 0.04 to 3 $\sqrt{\hbar\theta}$, is reported in the supplementary material)
while the memory parameter $\theta$ is kept constant at $\theta=1$. 

The comparison will be drawn between the dynamics driven by white noise $\xi_t$, \cref{eq:lindblad}, the OU process $X_\mathrm{OU}$, \cref{eq:qme_Xou}, and the associated $\Upsilon_\mathrm{OU}$-noise, \cref{eq:nM-QME_OUdriven}.
We remark that the dynamics of systems evolving under white noise and under a stochastic Hamiltonian modulated with an OU process are shown for comparison; the main focus is on the $\Upsilon_\mathrm{OU}$-noise driven SSE.

In \cref{fig:evol_trajs}, the difference between the typical behavior of the noisy trajectories, generated by the SSE approach, compared to the smooth propagation obtained from the random Hamiltonian with the OU processes can be appreciated, and it is a hallmark of the solution of the SSE.

\begin{figure*}
    \centering
    \includegraphics[width=0.85\linewidth]{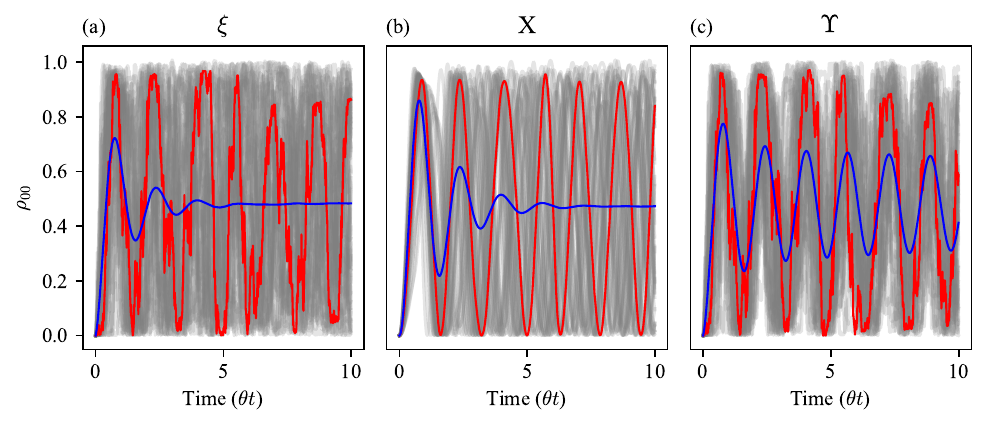}
    \caption{
    Comparison of the mean and single-trajectory evolution of the system under different driving stochastic processes.
    Swarms of 50 different stochastic trajectories (grey lines in transparency), the evolution of single trajectories (red lines), and the averaged evolution of the $\rho_{00}$ population---the observable under investigation (blue lines). 
    The stochastic force for each    unraveling: 
    in the upper left panel white noise SSE,
    central OU process fluctuations and
    right $\Upsilon_\mathrm{OU}$-noise-driven SSE. 
    Notice the difference in smoothness between a stochastic Hamiltonian propagation and the SSE propagations.
    Example for the the generic Hamiltonian $H_{\varepsilon,\Omega}$ system, with noise intensity $\gamma = 0.7 \sqrt{\hbar\theta}$ 
    and with noise operator $R=\sigma_x$.
    }
    \label{fig:evol_trajs}
\end{figure*}

\subsection{Diagonal noise}

When the noise is on the diagonal of the Hamiltonian in the local basis, i.e., $R=\sigma_z$, it physically represents fluctuations in the local state energies associated with
the instantaneous configurations of the bath.  
The rationale for this model is that bath-induced coupling between different local states is small relative to the interstates coupling
$\Omega$, and the bath affects mainly the local configurations. 

The average dynamics of the system affected by the different stochastic sources
are shown in \Cref{fig:PZ_hgen}. We chose parameters corresponding to slow relaxation so that coherent oscillations of the site population can be observed.  

In this setting, the $X_\mathrm{OU}$-driven dynamics show a slower relaxation of the beating, decaying with a single timescale. 
The behavior of the $\Upsilon_\mathrm{OU}$-driven system is similar to the white noise one,
with a slightly better robustness of the coherent oscillations to the relaxation.
This small effect gets more evident as the intensity of the noise increases (see supplementary material).
A second observation concerns the change in oscillation frequency, compared to the close-system dynamics.
While the $\xi$-driven mean dynamics does not show any frequency change,
the $\Upsilon_\mathrm{OU}$-driven systems show a slower oscillation. In contrast, the $X_\mathrm{OU}$-driven mean dynamics is characterized by a higher oscillation frequency, increasing with the noise intensity.

As a final comment, we observe that the initial time dynamics is characterized by a null derivative at time zero, which is expected as the physically correct behavior of a system interacting with an environment. This happens when the noise coupling operator is purely local, as in the present case, while it is not guaranteed in other cases. 

\begin{figure}
    \centering
    \includegraphics[width=0.85\linewidth]{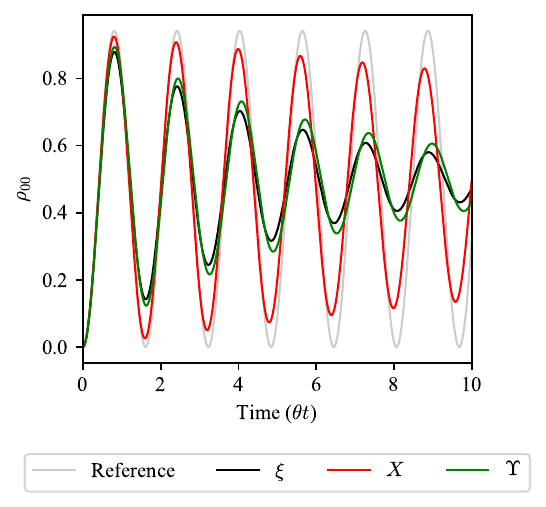}
    \caption{
    Dynamics of the system described by the Hamiltonian $H_{\varepsilon,\Omega}$, with the stochastic forces applied through the Pauli $\sigma_z$ operator.
    In light grey, the closed-system \schr equation evolution as the reference of the coherent beating frequency of the population.}
    \label{fig:PZ_hgen}
\end{figure}

\subsection{Noise on the coupling}

Within the setting of the noise on the off-diagonal elements of the site-basis frame,
$R=\sigma_x$,
we can work with a null Hamiltonian ($H_0=0$).
This allows us to study the effect of the sole dissipators, free of the intrinsic dynamics of the system,
as in the example in \cref{sec:EOMs}.
Using white noise as the source, the evolution with the noise applied through the Pauli $\sigma_x$ unravels a Pauli Master Equation.
The use of the other colored process and noise changes the dynamics.
In this setting, we clearly see the peculiar effect of the $\Upsilon_\mathrm{OU}$-noise as the driving of the dissipation channel.

\begin{figure*}
    \centering
    \includegraphics[width=0.92\linewidth]{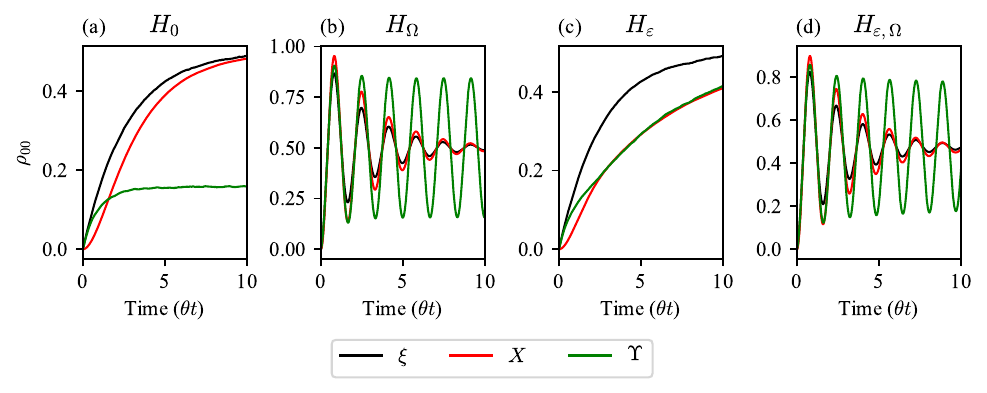}
    \caption{Dynamics of the system described by the different Hamiltonians, with the stochastic forces applied through the Pauli $\sigma_x$ operator.}
    \label{fig:px_Hall}
\end{figure*}

In \Cref{fig:px_Hall}a, the mean dynamics described by  \cref{eq:nM-QME_OUdriven},
in particular in the form \cref{eq:pauli-like_nM-QME},
is shown together with the mean dynamics resulting from the $\xi$-driven process, leading to a Lindblad dissipator,
and to the OU stochastic Hamiltonian dynamics, which bears only the term of correlation with the environment, without Lindblad-like components, \cref{eq:StochHamilt_notSDE_X}.
The noteworthy result is that the
process driven by the $\Upsilon_\mathrm{OU}$-noise reaches a different stationary state, 
while the other propagation reaches the expected equipartite distribution over the two sites.
The dynamics starts by following the same short-time behavior of the white noise case, but the stationary value of the site populations depends on the intensity of the white noise component (see supplementary material).

In \Cref{fig:px_Hall}b we show a similar behavior in the resonant coupled system ($H_{0,\Omega}$). 
In this case, the long-time populations are not constant, and an oscillating steady state is reached.

One could expect the environment correlation to affect mainly the short-time dynamics of the system.
Indeed, this is what happens using the OU process in all the instances of the Hamiltonian of the system.
In contrast, the dynamics with $\Upsilon_\mathrm{OU}$-driven dissipation starts with a transient similar to that of the $\xi$-driven one 
and then settles to a different stationary state.
On the other hand, 
the $X_\mathrm{OU}$ process-driven dynamics
start with a null derivative at time zero, typical of purely Hamiltonian dynamics.
This is particularly evident in the $H_{0,0}$ system, but the same effect is present in all the dynamics using $R=\sigma_x$. 
The different initial time behavior of the two noise-driven dissipations can be explained by the common Lindblad dissipator in the corresponding QMEs. 
This dissipator underlines an actual SDE component and not simply a stochastic process in the Hamiltonian, and clearly induces an exponential decay even at $t_0$. 

Up to this point, the two cases we have considered (\Cref{fig:px_Hall}a and b) are characterized by a symmetric Hamiltonian which commutes with the coupling operator to the stochastic component, $H_0$ being the trivial case.  
When we consider a system where the two sites are not degenerate, \Cref{fig:px_Hall}c, or the principal directions of the Hamiltonian are changed by a coupling term, \Cref{fig:px_Hall}d,
thus breaking the symmetry of the system, the dynamics changes.
The common and most important difference is that all dynamics in these cases reach the equipartite distribution; however, they do so following quite different transient dynamics. 

Let us first consider the simple addition of an energy asymmetry,
the Hamiltonian $H_{\varepsilon,0}$, where the two sites are not coupled
directly but through the interaction with the environment.
For the dynamics of the observable at short time,
the same considerations made above on the derivative at time zero and the similar behavior among the noise-driven and process-driven systems hold.
Then, for longer times,
in \Cref{fig:px_Hall}c for the selected noise intensity, the dynamics of the $\Upsilon_\mathrm{OU}$-driven system becomes comparable to the $X_\mathrm{OU}$-driven one,
eventually reaching equidistribution of the populations.

Depending on the intensity of the white noise component, different dynamics can be observed (see supplementary material).
At low intensity, the growth of the observable of the $\Upsilon_\mathrm{OU}$-driven system is slower than the $X_\mathrm{OU}$-driven one, yet closing in as $\gamma$ increases (see \cref{fig:px_he} in the supplementary material).
For higher noise intensity
the $\Upsilon_\mathrm{OU}$-driven system resembles more and more the $\xi$-driven system at short time and $X_\mathrm{OU}$-driven system at longer time.  
For $\gamma\gg1$, the white noise component prevails and the dissipation profile becomes indistinguishable from the Lindblad evolution.

When we add the coupling into the system Hamiltonian, the coherent evolution is characterized by two different relaxation regimes,
a faster decay at short times and then a slower decay (\Cref{fig:px_Hall}d), eventually driving the system to equipartition.
The correlation term seems to enhance the robustness of the coherent oscillations with respect to the noise, contrasting the Lindblad term of the dissipator.
As a result, the coherent oscillations are long-lived.

\subsection{Comparison with Redfield dynamics}

As derived in \Cref{sec:RED}, the Redfield equation with time-dependent coefficients can be used to approximate the average dynamics of the SSE drive by the $\Upsilon_\mathrm{OU}$-noise. In this section, the numerical solution of the Redfield model dynamics is investigated, and it shows good agreement with the SSE results for the observable under study.

\begin{figure*}
    \centering
    \includegraphics[width=0.85\linewidth]{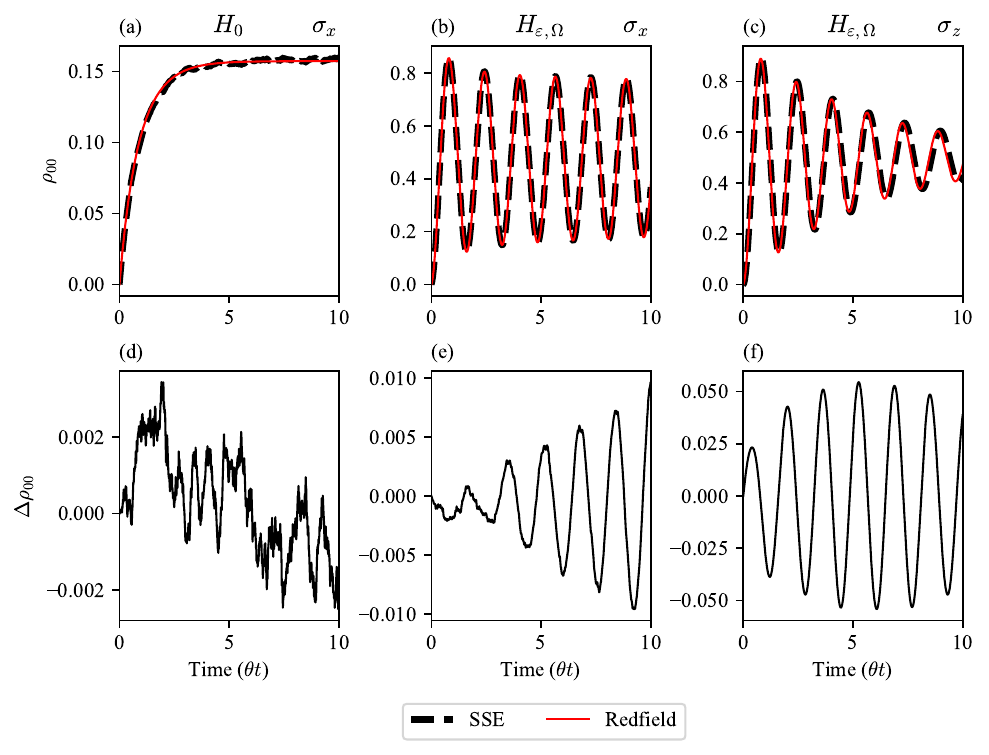}
    \caption{(a-c) Comparison of the $\rho_{00}$ dynamics,
    for the $H_{0}$ and $H_{\varepsilon,\Omega}$ Hamiltonians. 
    The dynamics computed by stochastic averaging are displayed as dashed bold black lines, and the Redfield model in red lines.
    (d-f) Absolute errors w.r.t.\ the SSE-obtained population.
    The noise is applied through $\sigma_x$ in (a,d) and (b,e),
    and through $\sigma_z$ in (c,f).}
    \label{fig:comparison_SSEvsRedfield}
\end{figure*}

In the simplest case of a null Hamiltonian, only the channel at zero frequency contributes to the relaxation. We recall that the associated rate vanishes at long time as shown in \cref{fig:timeCoeffiecientsRedfield}a. 
Looking at \cref{fig:comparison_SSEvsRedfield}a, we see that this approach correctly predicts
the dynamics of the null Hamiltonian system, the absolute error is shown in \cref{fig:comparison_SSEvsRedfield}d and it fluctuates with an amplitude that depends on the noisiness of the SSE trajectories average.
This is also true for coupled resonant systems.

When an energy difference is introduced in the system, like in 
the generic Hamiltonian $H_{\varepsilon,\Omega}$,
the difference between the Redfield dynamics and the average obtained from the SSE accumulates at early times because of slightly detuned coherent oscillations, \cref{fig:comparison_SSEvsRedfield}b.
At longer times, both dynamics converge to the same stationary states reaching equidistribution of populations.

The difference in the oscillation frequency is even more noticeable
in the case of noise on the sites ($R=\sigma_z$), see \cref{fig:comparison_SSEvsRedfield}c and \ref{fig:comparison_SSEvsRedfield}f.
The Redfield equation shows a slight speed-up of the oscillation, while the $\Upsilon_\mathrm{OU}$-derived QME slows them down depending on the intensity of the noise, see also \cref{fig:PZ_hgen}.

As a last note, one comment on the positivity of the dynamics.
Although there is good agreement in the numerical results, 
we want to stress that
the Redfield model does not ensure positivity of the dynamics, as pointed out in \Cref{sec:RED}.
In the results presented here, and for all parameters we used (see supplementary material),
including strong noise intensities $\gamma$, the numerical solution is always positive.

\section{Conclusions\label{Conclusions}}

Throughout this work, we have used the linear stochastic \schr equation unraveling approach to construct and understand
the quantum master equation generated by a colored environment. Using colored noise and processes as stochastic potentials, we obtained master equations that are open-form and not in Lindblad form, yet are positive definite by construction due to the averaging over pure-state trajectories.
This paves the way for using SSE as more than an efficient numerical strategy for known QMEs, but also for writing new non-trivial QMEs and exploring situations where an explicit QME cannot even be written.

From the analytical approach to the SSE using the \ito differential of the OU process as the stochastic drive, a less-known colored noise emerges, $\Upsilon_\mathrm{OU}$, in analogy to white noise arising from the use of a Wiener differential.
The definition and properties of this correlated noise are studied, and its effects on the dynamics of a two-level system are numerically investigated. 
We presented the normalized linear SSE driven by this noise as the unraveling of its open-form correlated QME.
In addition, we showed that different terms in the dissipator of the average dynamics correspond to either a Lindblad form unraveled by white noise SSE or correlation terms that are analogous to those obtained for the average dynamics of a stochastic Hamiltonian with Ornstein--Uhlenbeck fluctuations. 

Through the Redfield derivation, we gained insights into the peculiarity of the dynamics obtained using $\Upsilon_\mathrm{OU}$-colored noise, obtaining a closure model for the open term of the correlation of the system with the environment in the QME corresponding to the colored SSE.
Hence, though the SSE unravels a precise, but open form, master equation, \cref{eq:nM-QME_OUdriven},
we can think about these forms of Redfield QME, \cref{eq:RedfieldFrequenx,eq:RedfieldFreqFINAL}, as approximate models for it---an \textit{effective} QME for the SSE unraveling.


The analysis of the dynamics of a two-level system driven by colored SSE highlights important features.
When the environment is described as an $\Upsilon_\mathrm{OU}$-noise applied through the Pauli $\sigma_x$ operator to the interaction between the levels, a strong effect on the coherence time is observed.
When the system's Hamiltonian commutes with the noise operator, a different
asymptotic density matrix is reached, and the system does not reach equipartition of populations.
The steady state can be either static, for a null Hamiltonian, or an oscillating steady state for coupled levels.
The coherences are therefore maintained for these symmetric systems. 
The introduction of an asymmetry in the system, 
as an energy difference between the two levels,
results in reaching the equipartite distribution expected using symmetric relaxation operators. 
However, the coherences are long-lived, the dynamics showing two different regimes of relaxation 
with different time scales, i.e., a faster initial decay followed by a slower relaxation.
Although not analytically derivable from the SSE unraveling, the presence of different time scales is well explained by the Redfield master equation with time-dependent coefficients. The Redfield model also rationalizes the different stationary states and provides an insightful intuition on the effect of the noise in the system dynamics.

\section*{Supplementary materials}
The supplementary material contains additional information on the formal derivation of the correlation function of the correlated noise, the normalization condition constraints, 
and additional numerical results for varying noise intensities, both from trajectory-average and reduced model. 


\section*{Acknowledgments}

P.D.C. acknowledges the quantum computing resources provided by the Padua Quantum Computing and Simulation Center and the CINECA award under the \mbox{ISCRA} initiative, for the availability of high-performance computing resources and support.

P.D.C., G.G.G.\ and B.F.\ acknowledge funding from the European Union - NextGenerationEU, within the National Center for HPC, Big Data and Quantum Computing (Project no. CN00000013, CN1 Spoke 10: “Quantum Computing”).
B.F. and F.G.\ acknowledge the financial support by the Department of Chemical Sciences (DiSC) and the University of Padova with Project QA-CHEM (P-DiSC No. 04BIRD2021-UNIPD).

\bibliography{bibliography/reference}

\onecolumngrid
\newpage
\setcounter{page}{1}
\renewcommand{\thepage}{S\arabic{page}}

\renewcommand{\thefigure}{S\arabic{figure}}
\renewcommand{\thesection}{S\arabic{section}}
\renewcommand{\thesubsection}{S\arabic{subsection}}

\setcounter{figure}{0}
\setcounter{section}{0}
\setcounter{equation}{0}
\renewcommand{\theequation}{S\arabic{equation}}

\fontsize{11}{13.5}\selectfont
\onehalfspacing 
\setlength{\parindent}{0pt}

\begin{center}
    {\large\bfseries Supplementary Material \\Quantum trajectories and reduced dynamics in time-correlated environments}\\[0.5em]
    Pietro De Checchi$^1$, Federico Gallina$^{2}$, Barbara Fresch$^2$, Giulio G. Giusteri$^1$\\[1em]
    \textit{$^1$ Department of Mathematics, University of Padova, Via Trieste 63, Padova 35131,  Italy}\\
    \textit{$^2$ Department of Chemical Sciences, University of Padova, Via Marzolo 1, Padova 35131,  Italy}\\
\end{center}


\section{Properties of the noise associated with the Ornstein--Uhlenbeck process}
\label{app:CovY_computing}

The covariance of the $\Upsilon_\mathrm{OU}$ noise is obtained 
testing the two-time average against two test functions, $\phi,\psi\in\mathcal{C}^\infty_C$.
We first write the covariance in terms of derivatives of $(X_t)$:
\begin{equation}
    \text{cov}(\Upsilon_t,\Upsilon_s) = \mathbb{E}[\Upsilon_t\Upsilon_s] = \mathbb{E}[\dot{X}_t\dot{X}_s]
\end{equation}
and now we test it in the double-integral, which can be recast by integration by parts twice into
\begin{equation}
    \iint_{-\infty}^{+\infty} \mathbb{E}[\dot{X}_t\dot{X}_s]\phi_t\psi_s \diff t \diff s 
    =
    \iint_{-\infty}^{+\infty} \mathbb{E}[X_t\,X_s]\dot{\phi}_t\dot{\psi}_s \diff t \diff s
\end{equation}
the negative sign is not present as we integrate by parts twice and the boundary terms go to zero.
We have obtained a known quantity, the covariance of the Orstein-Uhlenbeck process, and by integrating twice by parts one last time:
\begin{equation}
    \iint_{-\infty}^{+\infty} \text{cov}(X_t, X_s) \dot{\phi}_t\dot{\psi}_s \diff t \diff s = \iint_{-\infty}^{+\infty} \frac{\diff}{\diff t}\left(\frac{\diff}{\diff s}\Big(\text{cov}(X_t, X_s)\Big)\right)\phi_t\psi_s \diff t \diff s 
\end{equation}
we find a way to obtain the covariance of the new noise.
Attention must be posed to the fact that the following computations have to be taken in the distributional sense, meaning they make sense upon integration.
\begin{subequations}
\begin{align}
    \text{cov}(\Upsilon_t,\Upsilon_s) &=  
    \frac{\diff}{\diff t}\left(\frac{\diff}{\diff s}\left(
    \frac{\gamma^2}{2\theta}\left( e^{-\theta |t-s|} -e^{-\theta (t+s)} \right)
    \right)\right) \\[2ex]
    &= \frac{\gamma^2}{2\theta} \frac{\diff}{\diff t}\left(
      \theta e^{-\theta |t-s|} \text{sign}(t-s)  + \theta e^{-\theta (t+s)} \right) \\[2ex]
    & = \frac{\gamma^2}{2\theta} \left(
      -\theta^2 e^{-\theta |t-s|} \text{sign}^2(t-s) +\theta e^{-\theta |t-s|} \text{sign}'(t-s)  - \theta^2 e^{-\theta (t+s)} 
      \right)
\end{align}
\end{subequations}

The expression of the covariance of the noise is obtained by recalling that this is valid upon the integration,
therefore 
we consider that the square of the sign is one, as the single point in zero brings no contribution,
and the derivative of the sign function can be written as
$ \text{sign}(x) = 2 \Theta(x) - 1 $,
where $\Theta(x)$ is the Heaviside function.
Then its derivative is not simply a Dirac's delta function, but twice that: 
$\text{sign}'(x)=2\delta_0(x)$.
Finally, the covariance of the $\Upsilon_\mathrm{OU}$ noise is: 
\begin{equation}
    \text{cov}(\Upsilon_t,\Upsilon_s) =  -\frac{\gamma^2 \theta}{2} e^{-\theta |t-s|} + \gamma^2\delta(t-s) 
\end{equation}

\newpage
\section{Constraint due to the normalization condition}\label{app:proofAntiHermit}

\begin{proposition}
\label{prop:Hermitian_BiR_toOUSSE_normalized}
    Using the stochastic differential $\diff{X_t}$ of an Ornstein--Uhlenbeck process $(X_t)_{t\geq0}$, 
    defined in \cref{eq:OrnsteinUhlenbeckProcess},
    as the noise source in the SSE \cref{eq:genericSSE} with deterministic operators $A$ and $B$,
    the martingale property of the norm   
    is ensured if and only if the 
    $B$ operator
    is an anti-Hermitian operator given by $B = iR$ with  $R\daggah=R$.
\end{proposition}

\begin{proof}

\begin{itemize}
    \item[($\Leftarrow$)] If $B=iR$, we compute the term to set to zero,
    \begin{equation}
        A\daggah+A -\theta X_t (-iR+iR) + \gamma^2R\daggah R = 0,
    \end{equation}
    that simplifies to 
            $ A\daggah + A + \gamma^2R\daggah R = 0 $,
    and the normalization term is obtained in the same fashion as for the white noise-driven SSE, \cref{eq:normal_A,eq:Aoperatorwithcorrection}:
    \begin{equation}
      A = -iH -\frac{1}{2}R\daggah R.
    \end{equation}
    
    \item[($\Rightarrow$)] 
    Given $(X_t)_{t\geq0}\in\mathbb{R},\mathcal{C}$ Gaussian r.v. with $\mathbb{E}[X_t]=0$, then
    $\forall\omega\,\exists t=t_c:\,X_{\omega,t_c}=0$.
    Then, $\forall\omega,\forall t=t_c$ the condition of normalization becomes again 
    \begin{equation}
       A\daggah + A + \gamma^2B\daggah B = 0 
    \end{equation}
    leading to the same normalization term as in \cref{eq:normal_A}.
    Since this normalization must hold true 
    at all times, then
    \begin{equation}
        \forall t\neq t_c ,  X_t\neq0\; \text{ then } \theta X_t(B\daggah+B)=0\,  \implies B\daggah = -B
    \end{equation}
\end{itemize}
\end{proof}

\newpage
\section{Additional numerical results}

In this section, we present additional numerical results illustrating the system dynamics under different stochastic driving mechanisms, together with a comparison to the dynamics obtained from the Redfield model with time-dependent coefficients, for noise intensities varied in the range $0.04$–$1.3\sqrt{\hbar\theta}$.
From \Cref{fig:PZ_hw} to \ref{fig:px_hGEN}, we compare the dynamics driven by white noise $\xi_t$, the Ornstein–Uhlenbeck process $X_{\mathrm{OU}}$, and $\Upsilon_{\mathrm{OU}}$ noise.
From \Cref{fig:sseVSredfield_H0} to \ref{fig:sseVSredfield_SIGMAZ}, we compare the numerical solution of the Redfield equation to the ensemble-averaged dynamics of the stochastic Schrödinger equation driven by $\Upsilon_{\mathrm{OU}}$ noise.





\begin{figure}[H]
    \centering
    \includegraphics[width=0.82\linewidth]{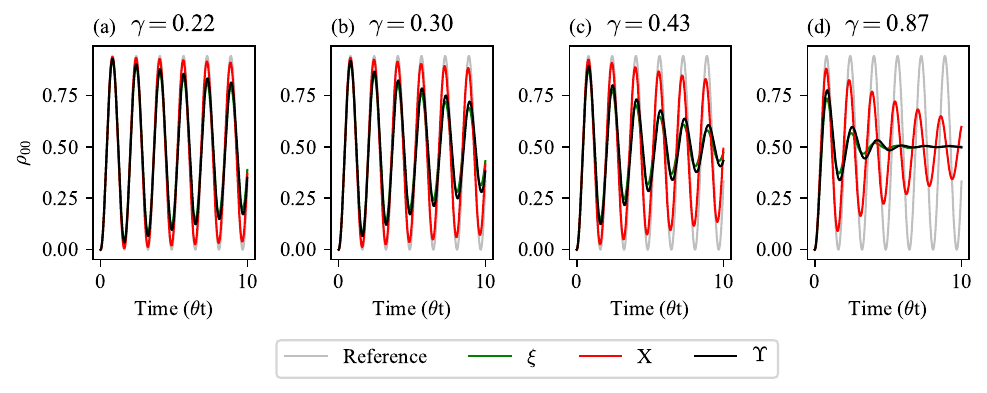}
    \caption{
    Few selected dynamics of the system described by the Hamiltonian $H_{\varepsilon,\Omega}$, with the stochastic forces applied through the Pauli $\sigma_z$ operator.}
    \label{fig:PZ_hw}
\end{figure}

\begin{figure}[H]
    \centering
    \includegraphics[width=0.82\linewidth]{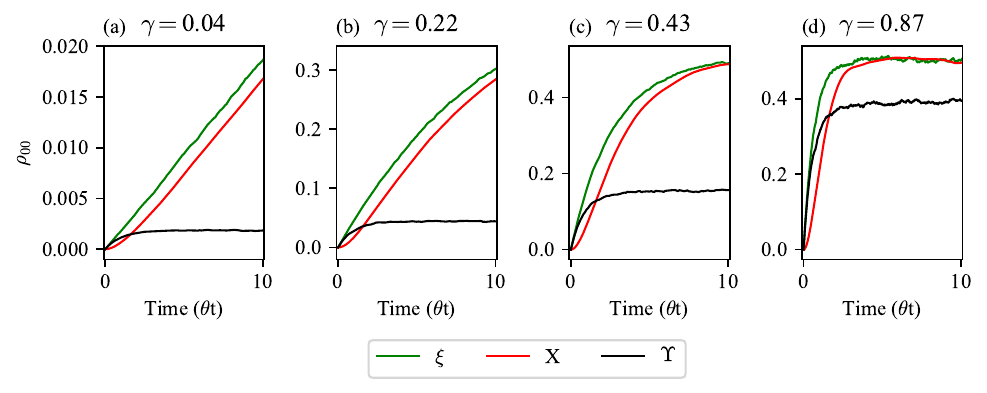}
    \caption{Few selected dynamics of the system described by the Hamiltonian $H_{0,0}$ with the stochastic forces applied through the Pauli $\sigma_x$ operator.}
    \label{fig:px_h0}
\end{figure}

\begin{figure}[H]
    \centering
    \includegraphics[width=0.82\linewidth]{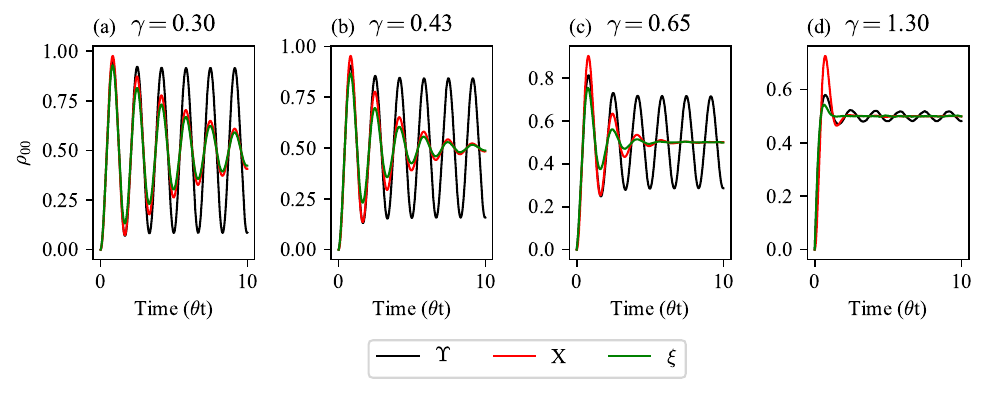}
    \caption{
    Few selected dynamics of the system described by the Hamiltonian $H_{0,\Omega}$ with the stochastic forces applied through the Pauli $\sigma_x$ operator.}
    \label{fig:px_hw}
\end{figure}

\begin{figure}[H]
    \centering
    \includegraphics[width=0.82\linewidth]{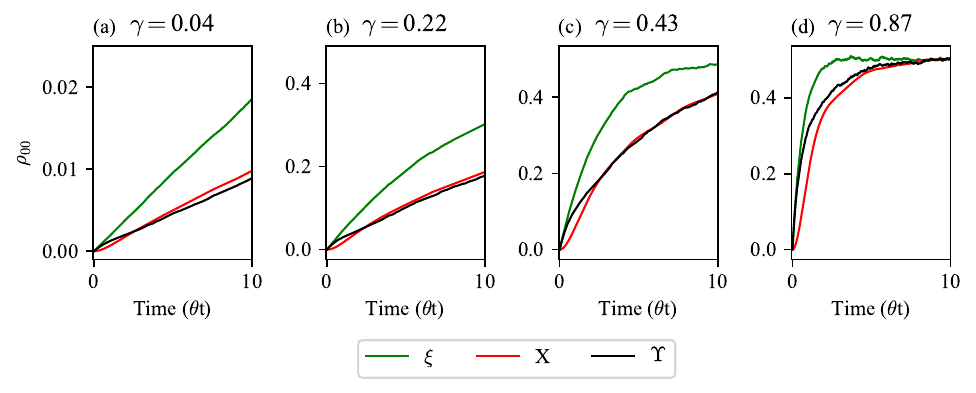}
    \caption{Few selected dynamics of the system described by the Hamiltonian $H_{\varepsilon,0}$ with the stochastic forces applied through the Pauli $\sigma_x$ operator.}
    \label{fig:px_he}
\end{figure}

\begin{figure}[H]
    \centering
    \includegraphics[width=0.82\linewidth]{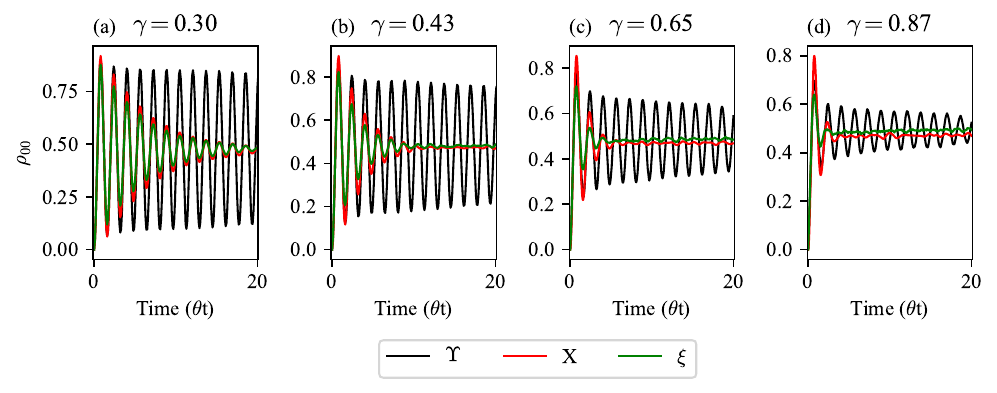}
    \caption{Few selected dynamics of the system described by the Hamiltonian $H_{\varepsilon,\Omega}$, with the stochastic forces applied through the Pauli $\sigma_x$ operator.}
    \label{fig:px_hGEN}
\end{figure}

\begin{figure}[H]
    \centering
    \includegraphics[width=0.65\linewidth]{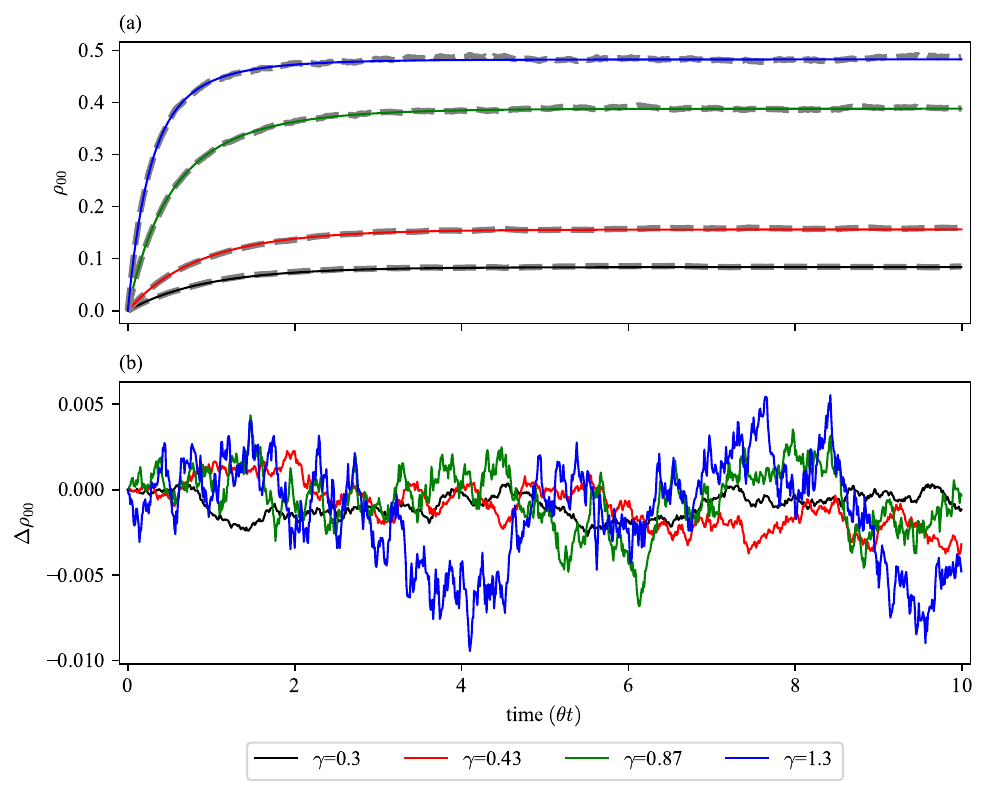}
    \caption{
    (a) Dynamics of the population $\rho_{00}$ of the $H_{0,0}$ system with $\sigma_x$ noise operator, obtained averaging over $\Upsilon_\mathrm{OU}$-driven SSE trajectories (solid lines) and from the Redfield with time-dependent coefficients (bold dashed grey lines), for different noise intensities.
    (b) Absolute error in the population dynamics.}
    \label{fig:sseVSredfield_H0}
\end{figure}

\begin{figure}[H]
    \centering
    \includegraphics[width=0.65\linewidth]{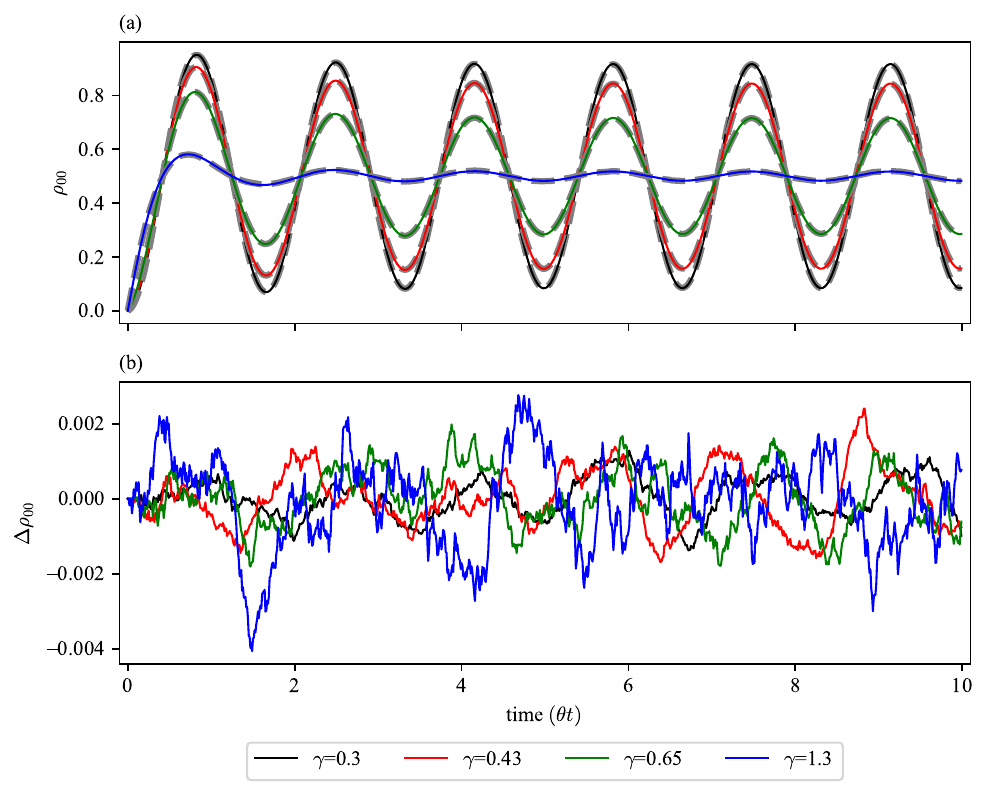}
    \caption{
    (a) Dynamics of the population $\rho_{00}$ of the $H_{0,\Omega}$ system with $\sigma_x$ noise operator, obtained averaging over $\Upsilon_\mathrm{OU}$-driven SSE trajectories (solid lines) and from the Redfield with time-dependent coefficients (bold dashed grey lines), for different noise intensities.
    (b) Absolute error in the population dynamics.}
    \label{fig:sseVSredfield_HW}
\end{figure}

\begin{figure}[H]
    \centering
    \includegraphics[width=0.65\linewidth]{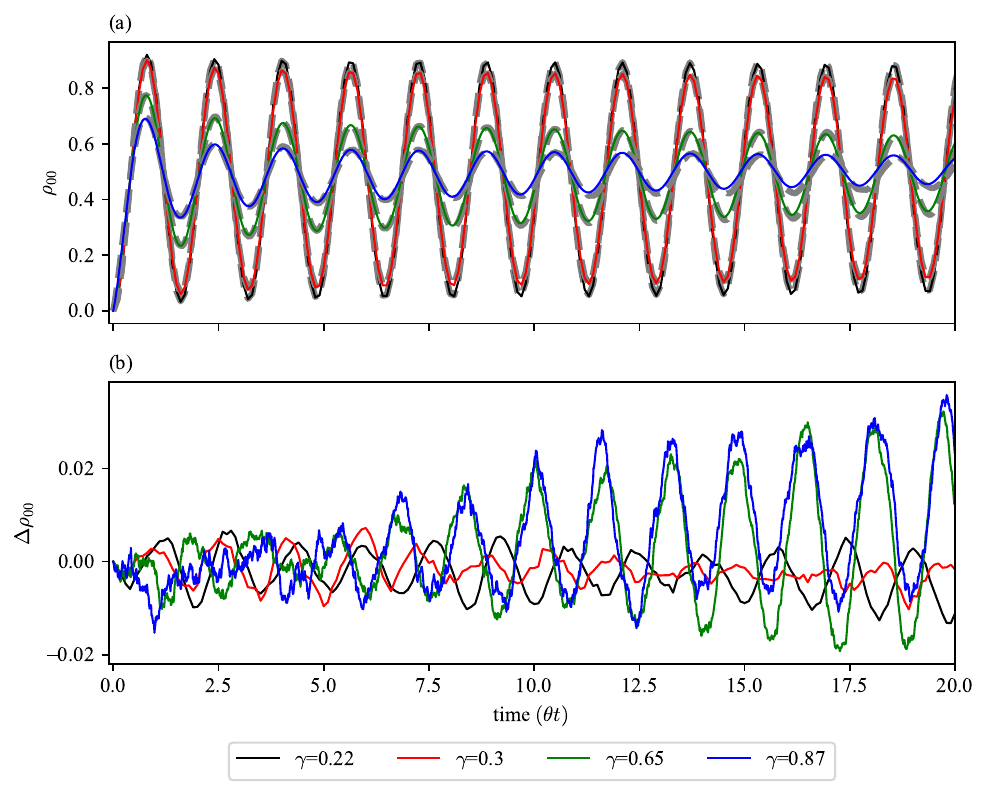}
    \caption{
    (a) Dynamics of the population $\rho_{00}$ of the $H_{\varepsilon,\Omega}$ system with $\sigma_x$ noise operator, obtained averaging over $\Upsilon_\mathrm{OU}$-driven SSE trajectories (solid lines) and from the Redfield with time-dependent coefficients (bold dashed grey lines), for different noise intensities.
    (b) Absolute error in the population dynamics.}
    \label{fig:sseVSredfield_HGEN}
\end{figure}

\begin{figure}[H]
    \centering
    \includegraphics[width=0.65\linewidth]{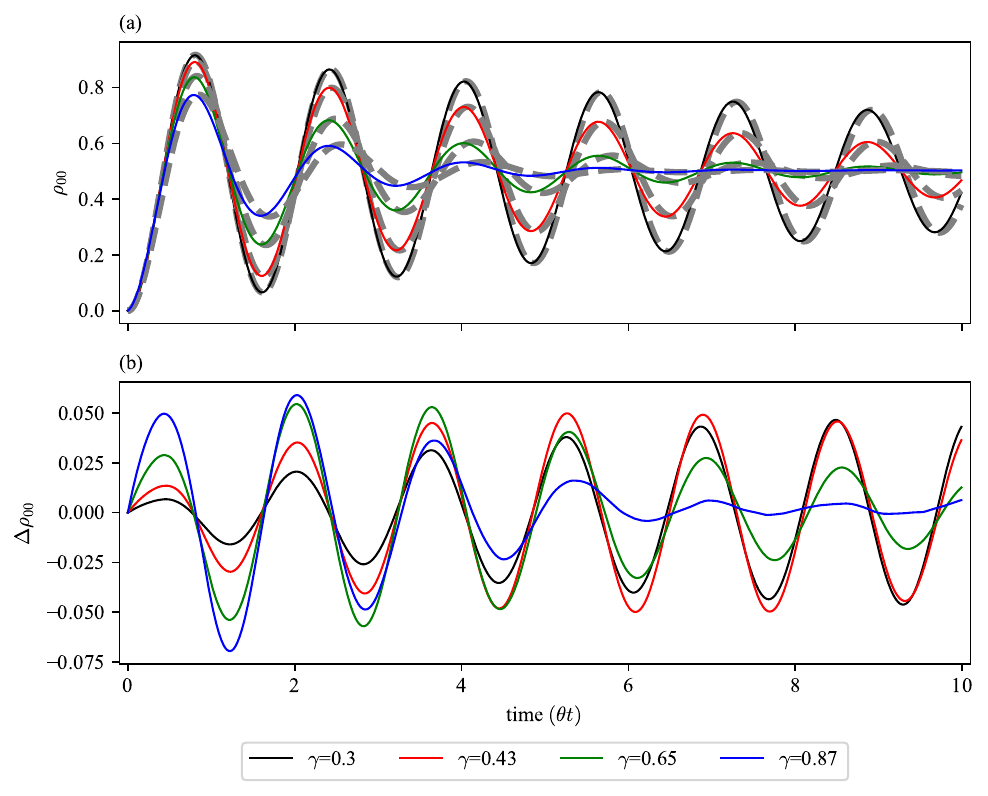}
    \caption{
    (a) Dynamics of the population $\rho_{00}$ of the $H_{\varepsilon,\Omega}$ system with $\sigma_z$ noise operator, obtained averaging over $\Upsilon_\mathrm{OU}$-driven SSE trajectories (solid lines) and from the Redfield with time-dependent coefficients (bold dashed grey lines), for different noise intensities.
    (b) Absolute error in the population dynamics.}
    \label{fig:sseVSredfield_SIGMAZ}
\end{figure}

\end{document}